\documentclass[11pt,a4paper]{article}
\usepackage{amssymb}
\usepackage{amsmath}
\usepackage{amsfonts}
\usepackage{bbm}
\usepackage{amsthm}
\usepackage{mathrsfs}
\usepackage{hyperref}
\usepackage{color}
\usepackage[margin=2.41cm]{geometry}
\usepackage[all,cmtip]{xy}
\usepackage[utf8]{inputenc}
\usepackage{graphicx}
\usepackage{varwidth}
\usepackage{comment}

\usepackage{upgreek}
\usepackage{rotating}

\usepackage{tikz}
\usetikzlibrary{shapes.geometric}

\definecolor{darkred}{rgb}{0.8,0.1,0.1}
\hypersetup{
     colorlinks=false,         % false: boxed links; true: colored links,false is default
     linkcolor=darkred,
     citecolor=blue,
}

\theoremstyle{plain}
\newtheorem{theo}{Theorem}[section]
\newtheorem{lem}[theo]{Lemma}
\newtheorem{propo}[theo]{Proposition}
\newtheorem{cor}[theo]{Corollary}

\theoremstyle{definition}
\newtheorem{defi}[theo]{Definition}

\newtheorem{exx}[theo]{Example}
\newenvironment{ex}
  {\pushQED{\qed}\exx}
  {\popQED\endexx}

\newtheorem{remm}[theo]{Remark}
\newenvironment{rem}
  {\pushQED{\qed}\remm}
  {\popQED\endremm}

\numberwithin{equation}{section}

\def\nn{\nonumber}

\def\bbK{\mathbb{K}}
\def\bbR{\mathbb{R}}
\def\bbC{\mathbb{C}}
\def\bbN{\mathbb{N}}

\def\Hom{\mathrm{Hom}}

\def\id{\mathrm{id}}

\def\1{I}
\def\oone{\mathbbm{1}}
\def\op{\mathrm{op}}

\def\Loc{\mathbf{Loc}}
\def\Lan{\operatorname{Lan}}
\def\Ran{\operatorname{Ran}}

\def\Set{\mathbf{Set}}
\def\Alg{\mathbf{Alg}}
\def\CAlg{\mathbf{CAlg}}
\def\Vec{\mathbf{Vec}}

\def\CC{\mathbf{C}}
\def\DD{\mathbf{D}}

\def\EE{\mathbf{E}}

\def\MM{\mathbf{M}}

\def\Cat{\mathbf{Cat}}
\def\OCat{\mathbf{OrthCat}}

\def\Seq{\mathbf{Seq}}

\def\Op{\mathbf{Op}}

\def\CCC{\mathfrak{C}}
\def\DDD{\mathfrak{D}}

\def\O{\mathcal{O}}
\def\P{\mathcal{P}}

\def\As{\mathsf{As}}
\def\Com{\mathsf{Com}}

\def\colim{\mathrm{colim}}

\newcommand\und[1]{\underline{#1}}
\newcommand\ovr[1]{\overline{#1}}

\DeclareMathOperator*{\Motimes}{\text{\raisebox{0.25ex}{\scalebox{0.8}{$\bigotimes$}}}}

\def\sk{\vspace{2mm}}

\makeatletter
\let\@fnsymbol\@alph
\makeatother

\title{%
Operads for algebraic quantum field theory
}

\author{%
Marco Benini$^{1,a}$, 
Alexander Schenkel$^{2,b}$\ and\
Lukas Woike$^{1,c}$\vspace{4mm}\\
{\small ${}^1$ Fachbereich Mathematik, Universit\"at Hamburg,}\\
{\small Bundesstr.~55, 20146 Hamburg, Germany.}\vspace{3mm}\\
{\small ${}^2$ School of Mathematical Sciences, University of Nottingham,}\\
{\small University Park, Nottingham NG7 2RD, United Kingdom.}\vspace{5mm}\\
{\small \begin{tabular}{ll}
Email: & ${}^a$~\texttt{benini@dima.unige.it}\\
& ${}^b$~\texttt{alexander.schenkel@nottingham.ac.uk}\\
& ${}^c$~\texttt{lukas.jannik.woike@uni-hamburg.de}\vspace{3mm}
\end{tabular}
}
}

\date{February 2020}

\begin{document}

\maketitle

\begin{abstract}
\noindent We construct a colored operad whose category of algebras is the category of algebraic quantum field theories. This is achieved by a construction that depends on the choice of a category, whose objects provide the operad colors, equipped with an additional structure that we call an orthogonality relation. This allows us to describe different types of quantum field theories, including theories on a fixed Lorentzian manifold, locally covariant theories and also chiral conformal and Euclidean theories. Moreover, because the colored operad depends functorially on the orthogonal category, we obtain adjunctions between categories of different types of quantum field theories. These include novel and interesting constructions, such as time-slicification and local-to-global extensions of quantum field theories. We compare the latter to Fredenhagen's universal algebra.
\end{abstract}

%\vspace{2mm}

\paragraph*{Report no.:} ZMP-HH/17-26, Hamburger Beitr\"age zur Mathematik Nr.\ 682

\paragraph*{Keywords:} algebraic quantum field theory, locally covariant quantum field theory, colored operads, 
change of color adjunctions, Fredenhagen's universal algebra

\paragraph*{MSC 2010:} 81Txx, 18D50

%\newpage 
%{\baselineskip=12pt
%%\setcounter{tocdepth}{2}
\tableofcontents
%}

%\bigskip

%\newpage

%%%%%%%%%%%%%%%%%%%%%%%%%%%%%%%%%%%%%%%%%%%%%%%%
%%%%%%%%%%%%%%%%%%%%%%%%%%%%%%%%%%%%%%%%%%%%%%%%

\section{\label{sec:intro}Introduction and summary}
Algebraic quantum field theory \cite{HaagKastler} is a conceptually clear axiomatic framework
to define and investigate quantum field theories on Lorentzian spacetimes
from a model-independent perspective.
Combining the core principles of quantum theory and relativity, it describes
a quantum field theory on a spacetime $M$ in terms of a coherent assignment
$M\supseteq U \mapsto \mathfrak{A}(U)$ of associative and unital
algebras to suitable spacetime regions.  $\mathfrak{A}(U)$ is interpreted 
as the algebra of quantum observables of the theory 
that can be measured in the region $U\subseteq M$.
Given two spacetime regions $U$ and $V$ such that 
$U\subseteq V \subseteq M$, there is an algebra homomorphism
$ \mathfrak{A}(U) \to\mathfrak{A}(V)$ mapping 
observables in the smaller region $U$ to the bigger region $V$.
These maps are required to be coherent in the sense that
$\mathfrak{A}$ is a pre-cosheaf. Given two causally disjoint regions
$U_1$ and $U_2$ of some $V\subseteq M$,
i.e.\ no causal curve in $V$ links $U_1$ and $U_2$, the elements of
$\mathfrak{A}(U_1)$ and $\mathfrak{A}(U_2)$ are required to commute
within $\mathfrak{A}(V)$. This crucial property is called the {\em Einstein causality axiom}
and it formalizes the physical principle that information cannot propagate
faster than the speed of light. The traditional framework \cite{HaagKastler} of 
algebraic quantum field theory on a fixed Lorentzian 
spacetime $M$ may be generalized and adapted in order to capture
also other flavors of quantum field theory. For example, one may also consider
the category of {\em all} (globally hyperbolic) Lorentzian manifolds $\Loc$,
which leads to the concept of locally covariant quantum field 
theory \cite{Brunetti,FewsterVerch}. Moreover, there exist
algebraic approaches to chiral conformal quantum field theory
\cite{Kawahigashi,Rehren,BDHcft} and Euclidean quantum field theory
\cite{Schlingemann}, where Lorentzian spacetimes are replaced 
respectively by intervals in the circle $\mathbb{S}^1$ or by Riemannian manifolds.
The Einstein causality axiom is modified in 
such scenarios to the requirement that observables associated to disjoint regions 
$U_1\cap U_2=\emptyset$ commute.
\sk

From a more abstract point of view, one observes that all these flavors of algebraic
quantum field theory have the following common features: There is a category
$\CC$ describing the ``spacetimes'' of interest. In this category there exists a distinguished subset
${\perp} \subseteq \mathrm{Mor}\,\CC \, \times\,\mathrm{Mor}\,\CC$, 
which we call an {\em orthogonality relation} (cf.\ Definition \ref{def:orthcat}), 
that is formed by certain pairs of morphisms 
$c_1 \stackrel{f_1}{\longrightarrow} c \stackrel{f_2}{\longleftarrow} c_2$ with the same target. 
For Lorentzian theories $\perp$ is characterized by causal disjointness 
and for chiral conformal or Euclidean theories by disjointness.
We call the pair $\ovr{\CC} := (\CC,\perp)$ consisting of
a category $\CC$ and an orthogonality relation $\perp$ an {\em orthogonal category}. 
The role of the orthogonal category $\ovr{\CC}$ is thus to specify
the flavor or type of quantum field theory one would like to study.
A quantum field theory on $\ovr{\CC}$ is then described
by a functor $\mathfrak{A} : \CC\to \Alg$ to the category
of associative and unital algebras, which satisfies the
{\em $\perp$-commutativity axiom}: 
For every $(f_1, f_2)\in\perp$, the induced commutator
\begin{flalign}
\big[ \mathfrak{A}(f_1) (-) , \mathfrak{A}(f_2)(-) \big]_{\mathfrak{A}(c)}^{} 
~:~\mathfrak{A}(c_1)\otimes\mathfrak{A}(c_2)~\longrightarrow~\mathfrak{A}(c)
\end{flalign}
is equal to the zero map. Hence, the category of quantum field theories
on $\ovr{\CC}$ is the full subcategory of the functor category $\Alg^\CC$
consisting of all functors satisfying the $\perp$-commutativity axiom.
\sk

The aim of this paper is to develop a more elegant and powerful description
of the category of all quantum field theories on $\ovr{\CC}$ by using techniques
from operad theory. Loosely speaking, operads are mathematical structures
that encode $n$-ary operations and their composition properties on an abstract level.
We shall construct, for each orthogonal category $\ovr{\CC}= (\CC,\perp)$,
a colored operad $\O_{\ovr{\CC}}$ whose category of algebras $\Alg_{\O_{\ovr{\CC}}}$ 
is the category of quantum field theories on $\ovr{\CC}$.
Hence, we succeed in identifying and extracting the abstract algebraic structures 
underlying algebraic quantum field theory.
It is worth emphasizing very clearly the key advantage of our novel operadic
perspective in comparison to the traditional  perspective:
In the functor approach, $\perp$-commutativity is a
{\em \underline{property}} that a functor $\mathfrak{A} : \CC\to\Alg$ may or may not
satisfy. In contrast to that, in our operadic approach every algebra 
over the colored operad $\O_{\ovr{\CC}}$ satisfies the $\perp$-commutativity axiom
because it is part of the {\em \underline{structure}} that is encoded in the operad $\O_{\ovr{\CC}}$.
Below we shall comment more on the crucial difference between property and structure
and on the resulting advantages of our operadic approach to algebraic quantum field theory.
\sk

We will also prove that the assignment $\ovr{\CC}\mapsto \O_{\ovr{\CC}}$
of our colored operad to an orthogonal category is functorial. This means that
for every orthogonal functor $F : \ovr{\CC}\to\ovr{\DD}$,
i.e.\ a functor preserving the orthogonality relations, there is an associated
colored operad morphism $\O_{F} : \O_{\ovr{\CC}} \to\O_{\ovr{\DD}}$.
This morphism induces an adjunction
\begin{flalign}\label{eqn:introadjunction}
\xymatrix{
{\O_{F}}_! \,:\, \Alg_{\O_{\ovr{\CC} }} ~\ar@<0.5ex>[r]&\ar@<0.5ex>[l]  ~\Alg_{\O_{\ovr{\DD}}} \,:\, \O_{F}^\ast
}
\end{flalign}
between the corresponding categories of algebras, which allows us
to relate the quantum field theories on $\ovr{\CC}$ to
those on $\ovr{\DD}$. Hence, our operadic approach
not only provides us with powerful tools to describe
the categories of quantum field theories of a fixed kind, 
but also introduces novel techniques to connect and compare different types of theories.
We will show that these adjunctions include some novel and interesting
constructions. For example: (1)~Using localizations of orthogonal categories,
we obtain adjunctions that should be interpreted physically as {\em time-slicification}.
This means that we can assign to theories that do not necessarily satisfy
the time-slice axiom --  a kind of dynamical law in Lorentzian quantum field theories --
theories that do. (2)~Using full orthogonal subcategory embeddings
$j :\ovr{\CC}\to\ovr{\DD}$, where $\ovr{\DD}$ are the ``spacetimes'' of interest
and $\ovr{\CC}$ particularly ``nice spacetimes'' in $\ovr{\DD}$,
we obtain an adjunction that should be interpreted physically as a {\em local-to-global extension}
of quantum field theories from $\ovr{\CC}$ to $\ovr{\DD}$. In spirit, this is
similar to Fredenhagen's universal algebra construction 
\cite{Fre1,Fre2,Fre3}, which is formalized as a left Kan extension of the functor 
underlying a quantum field theory \cite{Lang}. There is however one major difference:
Left Kan extensions in general do not preserve the $\perp$-commutativity {\em property} of a functor,
i.e.\ it is a priori unclear whether the prescription of \cite{Fre1,Fre2,Fre3,Lang} succeeds 
in defining a quantum field theory on $\ovr{\DD}$. 
In stark contrast to that, our operadic version of the local-to-global extension 
does always define quantum field theories on $\ovr{\DD}$
because the $\perp$-commutativity axiom is encoded as a {\em structure} in our colored operad. 
We will study this particularly important example of an adjunction in detail
and hope that it convinces the reader that our operadic framework is very useful for 
applications.
\sk

Our original motivation for developing an operadic approach 
to algebraic quantum field theory came from
{\em homotopical algebraic quantum field theory} 
\cite{BSShomotopy,BSfiberedingroupoids,BSSstack}.
This is a longer-term research program of two of us (M.B.~and A.S.),
whose goal is to combine algebraic quantum field theory
with techniques from homotopical algebra in order to
capture the crucial higher categorical structures that are 
present in quantum gauge theories.
While studying toy-models of such theories in \cite{BSfiberedingroupoids}, 
we observed that both the functorial structure and $\perp$-commutativity
are in general only realized up to homotopy in homotopical
algebraic quantum field theory. Combining our operadic approach
of the present paper with homotopical algebra 
provides a precise framework to describe algebraic quantum 
field theories up to coherent homotopies. This has been now
achieved in our follow-up paper \cite{BSWhoRan} and
by Yau in \cite{YauQFT}.
\sk

We would like to comment briefly on the relationship between our
approach and factorization algebras \cite{CostelloGwilliam}. From a superficial point of view,
the two frameworks appear very similar as they both employ colored operads
to encode the algebraic structures underlying quantum field theories.
However, the factorization algebra operad is rather different
from our family of colored operads $\O_{\ovr{\CC}}$ as it captures
only the multiplication of those observables that are localized 
in disjoint spacetime regions. In the follow-up paper 
\cite{BPSPFAAQFT}, we proved that these differences
are inessential in a Lorentzian setting with suitable additional hypotheses
by constructing an equivalence between certain types of algebraic quantum field theories 
and certain types of factorization algebras. See also \cite{GwilliamRejzner} for a comparison 
result that is based on BV quantization techniques.
\sk

The outline of the remainder of this paper is as follows:
In Section \ref{sec:prelim} we briefly recall some 
definitions and constructions from the theory of colored operads 
and their algebras.
In Section \ref{sec:QFToperads} we construct and study
our family of colored operads $\O_{\ovr{\CC}}$, where $\ovr{\CC}$ is an orthogonal category. 
In particular, we provide two different constructions of $\O_{\ovr{\CC}}$,
a direct definition and a presentation by generators and relations.
We then prove that the category $\Alg_{\O_{\ovr{\CC}}}$ of algebras 
over $\O_{\ovr{\CC}}$ is the category
of quantum field theories on $\ovr{\CC}$. We shall also provide
concrete examples of orthogonal categories $\ovr{\CC}$ that 
are relevant for algebraic quantum field theory.
In Section \ref{sec:QFTadjunctions} we study in detail
the properties of the adjunctions \eqref{eqn:introadjunction} that are induced
by orthogonal functors $F : \ovr{\CC}\to\ovr{\DD}$.
We show that these include interesting constructions
such as time-slicification and local-to-global extensions.
In Section \ref{sec:QFTconstructions} we compare our operadic local-to-global
extension to Fredenhagen's universal algebra construction, which is given by left
Kan extension of the functor underlying a quantum field theory. The general result is that,
whenever the left Kan extension yields a $\perp$-commutative functor,
then it coincides with our operadic construction. We shall provide
examples when this is the case, but also counterexamples
for which Fredenhagen's construction yields a functor that is not
$\perp$-commutative.

%%%%%%%%%%%%%%%%%%%%%%%%%%%%%%%%%%%%%%%%%%%%%%%%
%%%%%%%%%%%%%%%%%%%%%%%%%%%%%%%%%%%%%%%%%%%%%%%%

\section{\label{sec:prelim}Preliminaries}
In this section we briefly recall the necessary aspects of the 
theory of colored operads and their algebras.
All colored operads in this paper will be $\Set$-valued,
however we consider operad algebras with values
in any bicomplete closed symmetric monoidal category $\MM$.
A detailed presentation of the theory of colored operads can be found in 
\cite{WhiteYau,Yau,BergerMoerdijk,Gambino}, see also
\cite{LodayVallette,Fresse,Rezk} for the uncolored case $\CCC=\{\ast\}$.

\subsection{Colored operads}
\begin{defi}\label{def:operad}
A {\it colored operad} $\O$ consists of the following data: 
\begin{itemize}
\item[a)] a set of {\it colors} $\CCC$, 

\item[b)] for all $t \in \CCC$, $n \geq 0$, 
$\und{c} := (c_1, \dots, c_n) \in \CCC^n$, 
a set of {\em  operations} $\O\big(\substack{t\\\und{c}}\big)$,

\item[c)] for all $t \in \CCC$, $n \geq 0$, 
$\und{c} \in \CCC^n$, $\sigma \in \Sigma_n$, 
a map $\O(\sigma): \O\big(\substack{t\\\und{c}}\big) 
\to \O\big(\substack{t\\\und{c}\sigma}\big)$, 
where $\Sigma_n$ denotes the permutation group on $n$ letters 
and $\und{c}\sigma := (c_{\sigma(1)},\dots,c_{\sigma(n)})$, 

\item[d)] for all $t \in \CCC$, $m \geq 1$, $\und{a} \in \CCC^m$, 
$k_i \geq 0$, $\und{b}_i \in \CCC^{k_i}$, for $i = 1, \dots, m$, 
a {\em composition} map $\gamma: \O\big(\substack{t\\\und{a}}\big) \times 
\prod_{i=1}^m \O\big(\substack{a_i\\\und{b}_i}\big) 
\to \O\big(\substack{t\\(\und{b}_1,\dots,\und{b}_m)}\big)$, 

\item[e)] for all $t \in \CCC$, a {\em unit} element $\oone \in \O\big(\substack{t\\t}\big)$. 
\end{itemize}
These data are subject to the obvious permutation action, 
equivariance, associativity and unitality axioms, see e.g.\ \cite[Definition 11.2.1]{Yau}. 
A morphism of colored operads $\phi: \O \to \P$ consists of
\begin{itemize}
\item[a)] a map of the underlying sets of colors $\widetilde{\phi} : \CCC \to \DDD$, 

\item[b)] for all $t \in \CCC$, $n \geq 0$, $\und{c} \in \CCC^n$, 
a map $\phi: \O\big(\substack{t\\\und{c}}\big) 
\to \P\big(\substack{\widetilde{\phi}(t)\\\widetilde{\phi}(\und{c})}\big)$, 
where $\widetilde{\phi}(\und{c}) := (\widetilde{\phi}(c_1), \dots, \widetilde{\phi}(c_n))$, 
\end{itemize}
such that the permutation actions, compositions and units are preserved. 
We denote by $\Op$ the category of colored operads and
by $\Op_\CCC$ the category of colored operads over a fixed non-empty 
set of colors $\CCC$ (also called {\em $\CCC$-colored operads}), 
with morphisms acting as the identity on colors. 
\end{defi}

\begin{ex}\label{ex:Op}
The simplest example of an uncolored operad, i.e.\ $\CCC = \{\ast\}$ is a singleton,
is the {\em commutative operad} $\Com \in \Op$. 
The operations in arity $n$ are given by a singleton 
$\Com(n) := \{\ast\}$, 
for all $n \geq 0$. 
This already determines the operad structure. 
The single operation in $\Com(n)$ corresponds to the unique way 
of multiplying $n$ elements of a commutative, associative and unital 
algebra, cf.\ Example \ref{ex:Alg} below. 
\sk

Another basic example of an uncolored operad 
is the associative operad $\As \in \Op$. 
For $n \geq 0$, its set of $n$-ary operations 
$\As(n) := \Sigma_n$ 
is given by the permutations on $n$ letters. 
Each operation corresponds to one of the $\Sigma_n$-many (a priori
different) ways of multiplying $n$ elements of an associative and unital algebra, 
cf.\ Example \ref{ex:Alg} below. 
The $\Sigma_n$-action on $\As(n)$ is given 
by the right action of $\Sigma_n$ on itself. 
For $m \geq 1$ and $k_i \geq 0$, for $i=1,\dots,m$, 
the composition 
\begin{flalign}
\sigma(\sigma_1,\dots,\sigma_n) ~:=~ \sigma\langle k_{\sigma^{-1}(1)},\dots, k_{\sigma^{-1}(m)} \rangle~
(\sigma_1\oplus\cdots\oplus\sigma_m)
\end{flalign} 
of $\sigma \in \As(m)$ with $\sigma_i \in \As(k_i)$, for $i=1,\dots,m$, 
is the group multiplication in $\Sigma_{k_1 + \dots + k_m}$ 
of the block permutation $\sigma\langle k_{\sigma^{-1}(1)},\dots, k_{\sigma^{-1}(m)} \rangle$
induced by $\sigma$ and the sum permutation $\sigma_1\oplus\cdots\oplus\sigma_m$ induced by 
the $\sigma_i$. The unit $e \in \As(1)$ is the unique element of $\Sigma_1$. 
The combinatorics of permutations implies the
equivariance, associativity and unitality axioms, i.e.\ $\As$ is an operad. 
\end{ex}

Forgetting the permutation action, 
composition and unit in Definition \ref{def:operad},
one obtains the category of colored (non-symmetric) sequences $\Seq$
and the categories $\Seq_\CCC$ of colored sequences over a fixed non-empty 
set of colors $\CCC$ (also called {\em $\CCC$-colored sequences}).
We have an evident forgetful functor $U: \Op_\CCC \to \Seq_\CCC$. 
\begin{theo}\label{theo:freeoperad}
Let $\CCC$ be a non-empty set of colors. There exists an adjunction 
\begin{flalign}
\xymatrix{
F \,:\, \Seq_\CCC ~\ar@<0.5ex>[r]&\ar@<0.5ex>[l] ~\Op_\CCC \,:\, U\quad.
}
\end{flalign}
The left adjoint $F : \Seq_\CCC \to \Op_\CCC$ is
called the {\em free $\CCC$-colored operad functor}.
\end{theo}

\begin{rem}\label{rem:freeoperad}
There exist various (equivalent) models for the free $\CCC$-colored operad functor,
see e.g.\  \cite{BergerMoerdijk,Rezk,WhiteYau,Yau}. We do not have to describe 
any particular model in full detail, because existence of the free $\CCC$-colored operad 
functor is sufficient for proving our results.
In order to introduce the reader to the graphical notations used 
in Section \ref{subsec:gensrels}, it is useful to sketch very briefly the 
construction of free colored operads by rooted trees, 
see e.g.\ \cite[Chapters 19 and 20]{Yau} for the details.
Consider a $\CCC$-colored sequence $X\in \Seq_\CCC$,
whose elements $o\in X\big(\substack{t \\ \und{c}}\big)$ are interpreted
as the generating operations from $\und{c} = (c_1,\dots,c_n)$ to $t$.
We may visualize these generating operations by $\CCC$-colored corollas
\begin{flalign}
\begin{tikzpicture}[cir/.style={circle,draw=black,inner sep=0pt,minimum size=2mm},
        poin/.style={circle, inner sep=0pt,minimum size=0mm}]
\node[poin] (Mout) [label=above:{\small $t$}] at (0,0.8) {};
\node[poin] (Min1) [label=below:{\small $c_1$}] at (-0.5,0) {};
\node[poin] (Min2) [label=below:{\small $\, c_n$}] at (0.5,0) {};
\node[poin] (Min0) [label=below:{\small $\cdots$}] at (0,0) {};
\node[poin] (V)  [label=left:{\small $o$}] at (0,0.4) {};
\draw[thick] (Min1) -- (V);
\draw[thick] (Min2) -- (V);
\draw[thick] (V) -- (Mout);
\end{tikzpicture}
\end{flalign}
The operations of the free $\CCC$-colored operad $F(X)\in \Op_\CCC$
are then given by $\CCC$-colored planar rooted trees, together with
a decoration of the inner vertices by generating operations and a 
permutation of the input profile. The full definition is 
given in \cite[Definitions 19.1.1 and 20.1.2]{Yau}. 
The operations of $F(X)$ admit a convenient
graphical representation, for example 
\begin{flalign}
\begin{tikzpicture}[cir/.style={circle,draw=black,inner sep=0pt,minimum size=2mm},
        poin/.style={circle, inner sep=0pt,minimum size=0mm}]
\node[poin] (Mout) [label=above:{\small $t$}] at (0,0.8) {};
\node[poin] (V1)  [label=left:{\small $o_2$}] at (0,0.4) {};
\node[poin] (Mmid1)  at (-0.3,0) {};
\node[poin] (Mmid2) [label=right:{\small $o_1$}]  at (0.3,0) {};
\node[poin] (Min1)   at (-0.6,-0.4) {};
\node[poin] (Min2)  at (0,-0.4) {};
\node[poin] (Min3)  at (0.6,-0.4) {};
\node[poin] (Mperm1) [label=below:{\small $c_1$}]  at (-0.6,-0.8) {};
\node[poin] (Mperm2) [label=below:{\small $c_2$}]   at (0,-0.8) {};
\node[poin] (Mperm3) [label=below:{\small $c_3$}]   at (0.6,-0.8) {};
\draw[thick] (Min1) -- (V1);
\draw[thick] (Mmid2) -- (V1);
\draw[thick] (V1) -- (Mout);
\draw[thick] (Min2) -- (Mmid2);
\draw[thick] (Min3) -- (Mmid2);
\draw[thick] (Mperm1) -- (Min3);
\draw[thick] (Mperm2) -- (Min1);
\draw[thick] (Mperm3) -- (Min2);
\end{tikzpicture}
\end{flalign}
with $o_1 \in X\big(\substack{b \\ (c_3,c_1)}\big) $
and $o_2\in X\big(\substack{t \\ (c_2,b)}\big)$ defines an operation in 
$F(X)\big(\substack{ t \\ (c_1,c_2,c_3)}\big)$. This intuitive graphical notation 
will be used in Section \ref{subsec:gensrels} to obtain
a presentation of our algebraic quantum field theory operad by generators and relations.
\end{rem}

The category $\Seq_\CCC$ is obviously complete and cocomplete. 
By \cite[Theorem 3.8]{PavlovScholbach}, the same is true 
for the category $\Op_\CCC$. 
\begin{propo}\label{propo:colimitsinOp}
The category $\Op_\CCC$ is complete and cocomplete.
Moreover, the forgetful functor $U: \Op_\CCC \to \Seq_\CCC$
creates all small limits, filtered colimits and reflexive coequalizers.
\end{propo}

\begin{rem}\label{rem:genrel}
We will later construct $\CCC$-colored operads from generators and relations. 
This is possible due to the free $\CCC$-colored operad functor 
and existence of coequalizers in $\Op_\CCC$. 
Let us illustrate this construction:
Given $X \in \Seq_\CCC$, consider
the free $\CCC$-colored operad $F(X) \in \Op_\CCC$.
Relations in $F(X)$ are implemented via a pair of parallel 
$\Seq_\CCC$-morphisms $r_1,r_2 : R \rightrightarrows F(X)$.
(We suppress here and in the following the forgetful functor 
$U: \Op_\CCC \to \Seq_\CCC$ from our notation.)
The adjunction in Theorem \ref{theo:freeoperad} defines a pair of 
parallel $\Op_\CCC$-morphisms $r_1,r_2 : F(R)\rightrightarrows F(X)$ 
(still denoted by the same symbols), 
whose coequalizer is the $\CCC$-colored operad generated by $X\in \Seq_\CCC$
subject to the relations $r_1,r_2 : R \rightrightarrows F(X)$. 
\end{rem}

\subsection{\label{subsec:Alg}Algebras over colored operads}
Algebras over a colored operad $\O$
should be interpreted as concrete realizations of the abstract 
operations encoded by $\O$. 
For the definition, let us fix once and for all a bicomplete 
closed symmetric monoidal category 
$\MM$. The symmetric monoidal structure 
of $\MM$ will be denoted by $(\otimes, I, \tau)$, 
where $\otimes: \MM \times \MM \to \MM $ is the monoidal product, 
$I \in \MM$ the monoidal unit and $\tau$ the symmetric braiding. 
Because $\MM$ is by assumption bicomplete, it is tensored and cotensored over $\Set$. 
With a slight abuse of notation, we shall denote the $\Set$-tensoring by $S\otimes m :=\coprod_{s\in S} m $,
for all $S\in \Set$ and $m\in\MM$. 
Relevant examples are the category of sets $\Set$ (with Cartesian product) and 
the category of $\bbK$-vector spaces $\Vec_\bbK$ (with the usual tensor product). 
 
\begin{defi}\label{def:Alg}
Let $\O$ be a $\CCC$-colored operad. An {\em algebra $A$ over $\O$ with values in $\MM$}
(or briefly, an {\em $\O$-algebra in $\MM$}) consists of the following data: 
\begin{itemize}
\item[a)] for all $t \in \CCC$, an object $A_t \in \MM$, 

\item[b)] for all $t \in \CCC$, $n \geq 0$, $\und{c} \in \CCC^n$, 
an $\MM$-morphism $\alpha: \O\big(\substack{t \\ \und{c}}\big) \otimes A_{\und{c}} 
\to A_t$ (called {\em $\O$-action}), 
where $A_{\und{c}} := \bigotimes_{i=1}^n A_{c_i}$. 
We write $\alpha(o): A_{\und{c}} \to A_t$ for 
the component of $\alpha$ at $o \in \O\big(\substack{t \\ \und{c}}\big)$. 
\end{itemize}
These data are subject to the obvious equivariance, associativity 
and unitality axioms with respect to the permutation action, 
composition and unit of $\O$, see e.g.\ \cite[Definition 13.2.3]{Yau}. 
A morphism of $\O$-algebras $\kappa : A \to B$ consists of a family of $\MM$-morphisms 
$\kappa : A_t \to B_t$, for all $t\in \CCC$, which is compatible with the
$\O$-actions, i.e.\ $\kappa \,\alpha^A = \alpha^B\, (\id\otimes \Motimes_{i=1}^n\kappa )$.
We denote the category of $\O$-algebras in $\MM$ by $\Alg_\O(\MM)$. 
\end{defi}

\begin{ex}\label{ex:Alg}
With reference to Example \ref{ex:Op}, we note that
$\Alg_{\Com}(\MM) = \CAlg(\MM)$ is the category of commutative, associative
and unital algebras in $\MM$ and that
$\Alg_{\As}(\MM) = \Alg(\MM)$ is the category of associative
and unital algebras in $\MM$.
\end{ex}

Let us recall some relevant categorical properties of $\Alg_\O(\MM)$, 
for $\O\in\Op_\CCC$ any $\CCC$-colored operad. Denoting
by $\MM^\CCC$ the category of functors from $\CCC$ (seen as a discrete category)
to $\MM$, there is an evident functor $U_\O: \Alg_\O(\MM) \to \MM^\CCC$
which forgets the $\O$-action. This functor has a left adjoint, which may
be computed in terms of coends, see e.g.\ \cite{Fosco,MacLane} for a detailed discussion
of (co)ends and also \cite{FuchsSchweigert}  for their applications to quantum field theory. 
\begin{theo}\label{theo:freealgebraoperad}
Let $\O$ be a $\CCC$-colored operad. There exists an adjunction 
\begin{flalign}
\xymatrix{
F_\O \,:\, \MM^\CCC ~\ar@<0.5ex>[r]&\ar@<0.5ex>[l] ~\Alg_{\O}(\MM) \,:\, U_\O \quad.
}
\end{flalign}
Explicitly, for $X\in\MM^\CCC$, the free $\O$-algebra 
$F_\O(X)$ consists of the following data: 
\begin{itemize}
\item for all $t\in\CCC$, the $\MM$-object
\begin{flalign}\label{eqn:freealgformula}
F_\O(X)_t ~:=~ \int^{\und{a}\in\Sigma_\CCC} \O\big(\substack{t\\\und{a}}\big) \otimes X_{\und{a}} 
\end{flalign} 
defined by the coend over the permutation groupoid\footnote{The objects
of $\Sigma_\CCC$ are all tuples $\und{c} = (c_1,\dots,c_n)$ of colors (of arbitrary finite 
length $n\geq 0$) and the morphisms are all right permutations 
$\sigma : \und{c}\to \und{c}\sigma$.} 
$\Sigma_\CCC$ of the set of colors $\CCC$,

\item for all $t\in\CCC$, $n \geq 0$, $\und{c}\in\CCC^n$, 
the $\MM$-morphism 
\begin{flalign}
\gamma_X\, :\,  \O\big(\substack{t\\ \und{c}}\big) \otimes F_\O(X)_{\und{c}} 
~\longrightarrow~ F_\O(X)_t 
\end{flalign}
induced by the composition $\gamma$ of $\O$ and the universal property coends. 
\end{itemize}
The unit $\oone_X: X \to F_\O(X)$ of the adjunction is defined by the unit $\oone$ of $\O$ 
and the counit $\alpha: F_\O(A) \to A$ is given by the $\O$-action of $A \in \Alg_\O(\MM)$. 
\end{theo}
The following statement is proven in \cite[Proposition 4.2.1]{WhiteYau},
see also \cite[Proposition 1.3.6]{Fresse} for a more concrete proof 
for the case of uncolored operads, which generalizes 
to colored operads in a straightforward way. 
\begin{propo}\label{propo:colimitsinAlg}
Let $\O$ be a $\CCC$-colored operad.
The category $\Alg_\O(\MM)$ is complete and cocomplete.
Moreover, the forgetful functor $U_\O : \Alg_\O(\MM) \to \MM^\CCC$
creates all small limits, filtered colimits 
and coequalizers which are reflexive in the category $\MM^\CCC$.
\end{propo}

Given any $\Op$-morphism $\phi: \O \to \P$, with underlying map
of sets of colors denoted by $\widetilde{\phi}:\CCC\to\DDD$,
there exists a pullback functor of algebras  
\begin{flalign}\label{eqn:fphiastAlg}
\phi^\ast \,:\, \Alg_\P(\MM)~\longrightarrow~\Alg_\O(\MM)\quad.
\end{flalign}
Concretely, for a $\P$-algebra $B\in\Alg_\P(\MM)$, define the $\O$-algebra 
$\phi^\ast(B)\in\Alg_\O(\MM)$ by $\phi^\ast(B)_t := B_{\widetilde{\phi}(t)} \in \MM$, 
for all $t \in \CCC$, and $\O$-action $\alpha^{\phi^\ast(B)}$ given by the composition of
\begin{flalign}
\xymatrix@C=3.5em{
\O\big(\substack{t\\\und{c}}\big) \otimes B_{\widetilde{\phi}(\und{c})}  \ar[r]^-{\phi \otimes \id}&
 \P\big(\substack{\widetilde{\phi}(t)\\\widetilde{\phi}(\und{c})}\big) \otimes B_{\widetilde{\phi}(\und{c})} \ar[r]^-{\alpha^B} & B_{\widetilde{\phi}(t)}
}\quad,
\end{flalign}
for all $t \in \CCC$, $n \geq 0$ and $\und{c} \in \CCC^n$. 
\begin{theo}\label{theo:coloradjunctionalgebras}
For any $\Op$-morphism $\phi: \O \to \P$, there exists an adjunction
\begin{flalign}\label{eqn:coloradjunctionalgebras}
\xymatrix{
\phi_! \,:\, \Alg_\O(\MM) ~\ar@<0.5ex>[r]&\ar@<0.5ex>[l] ~\Alg_\P(\MM) \,:\, \phi^\ast\quad.
}
\end{flalign}
\end{theo}
\begin{proof}
This is an application of the Adjoint Lifting Theorem, see e.g.\ \cite[Chapter 4.5]{handbook2}.
Concretely, we have the following diagram of categories and functors
\begin{flalign}\label{eqn:adjointliftingdiagram}
\xymatrix{
\ar@<0.5ex>[d]^-{U_\O} \Alg_\O(\MM) ~&&~ \ar[ll]_-{\phi^\ast}\Alg_\P(\MM) \ar@<0.5ex>[d]^-{U_\P} \\
\ar@<0.5ex>[u]^-{F_\O} \MM^\CCC \ar@<0.5ex>[rr]^-{\widetilde{\phi}_!}~&&~\ar@<0.5ex>[ll]^-{\widetilde{\phi}^\ast} \MM^\DDD\ar@<0.5ex>[u]^-{F_\P}
}
\end{flalign}
The vertical adjunctions have been established 
in Theorem \ref{theo:freealgebraoperad}. The bottom horizontal adjunction 
is given by pullback $\widetilde{\phi}^\ast: \MM^\DDD \to \MM^\CCC$ and left Kan extension
$\widetilde{\phi}_!: \MM^\CCC \to \MM^\DDD$ along the underlying map of sets of 
colors $\widetilde{\phi} : \CCC \to \DDD$.
By definition of the pullback functor \eqref{eqn:fphiastAlg}, one has 
$U_\O \, \phi^\ast = \widetilde{\phi}^\ast  \, U_\P$, hence the desired 
left adjoint $\phi_!$ exists by \cite[Theorem 4.5.6]{handbook2}.
\end{proof}
The proof of the adjoint lifting theorem in \cite[Chapter 4.5]{handbook2}
provides an explicit construction of $\phi_!(A) \in \Alg_\P(\MM)$ in terms 
of a reflexive coequalizer, for any $A \in \Alg_\O(\MM)$.
Suppressing again all forgetful functors, the relevant pair of parallel morphisms 
in $\Alg_\P(\MM)$ is of the form
\begin{flalign}\label{eqn:colcoequal}
\xymatrix{
F_\P \, \widetilde{\phi}_! \, F_\O(A) \ar@<0.5ex>[r]^-{\partial_0}  \ar@<-0.5ex>[r]_-{\partial_1} ~&~  
F_\P \, \widetilde{\phi}_!(A)\quad,
}
\end{flalign}
where  $\partial_0$ and $\partial_1$ are described as follows.
The $\Alg_\P(\MM)$-morphism $\partial_0$ is defined 
by the counit $\alpha : F_\O(A)\to A$ of the adjunction $F_\O \dashv U_\O$ (cf.\ 
Theorem \ref{theo:freealgebraoperad}), i.e.\
\begin{flalign}\label{eqn:partial0}
\xymatrix@C=4em{
F_\P \, \widetilde{\phi}_! \, F_\O(A) ~\ar[rr]^-{\partial_0 \,:= \, F_\P \, \widetilde{\phi}_!(\alpha)}&&~
F_\P \, \widetilde{\phi}_!(A)
}\quad.
\end{flalign}
The $\Alg_{\P}(\MM)$-morphism $\partial_1$ is defined by the units and counits of the adjunctions in
\eqref{eqn:adjointliftingdiagram} via
\begin{flalign}\label{eqn:partial1}
\xymatrix@C=4em{
\ar[d]_-{\text{unit of $\widetilde{\phi}_!\dashv \widetilde{\phi}^\ast$}} F_\P \, \widetilde{\phi}_! \, F_\O (A) 
\ar[rr]^-{\partial_1} ~&&~ F_\P \, \widetilde{\phi}_!(A)\\
\ar[d]_-{\text{unit of $F_\P \dashv U_\P$}} F_\P \, \widetilde{\phi}_! \, F_\O \, \widetilde{\phi}^\ast\, \widetilde{\phi}_!(A)
 ~&&~ F_\P \, F_\P \, \widetilde{\phi}_!(A) \ar[u]_-{\text{counit of $F_\P \dashv U_\P$}} \\
F_\P \, \widetilde{\phi}_! \, F_\O \, \widetilde{\phi}^\ast \, F_\P \, \widetilde{\phi}_!(A) \ar[rr]_-{\text{counit of $F_\O \dashv U_\O$}} 
~&&~ F_\P \, \widetilde{\phi}_! \, \widetilde{\phi}^\ast \, F_\P \, \widetilde{\phi}_!(A) \ar[u]_-{\text{counit of $\widetilde{\phi}_!\dashv \widetilde{\phi}^\ast$}}
}
\end{flalign}
Notice that \eqref{eqn:colcoequal} is reflexive with reflector 
defined by the unit $\oone_A: A \to F_\O(A)$ of $F_\O \dashv U_\O$  (cf.\ 
Theorem \ref{theo:freealgebraoperad}), i.e.\ 
\begin{flalign}\label{eqn:reflectorrrrr}
\xymatrix@C=4em{
F_\P \, \widetilde{\phi}_!(A) \ar[rr]^-{F_\P \, \widetilde{\phi}_!(\oone_A) } ~&&~ F_\P \, \widetilde{\phi}_! \, F_\O(A)
}\quad.
\end{flalign}
Proposition~\ref{propo:colimitsinAlg} implies that the colimit of the reflexive
pair of  parallel morphisms \eqref{eqn:colcoequal} is created by 
the forgetful functor. Summing up, we obtain
\begin{propo}\label{propo:changeofcolorcoeq}
The left adjoint functor $\phi_!$ of Theorem \ref{theo:coloradjunctionalgebras}
maps an $\O$-algebra $A \in \Alg_\O(\MM)$ to the $\P$-algebra
given by the reflexive coequalizer
\begin{flalign}
\phi_!(A) ~=~ \colim\Big(\xymatrix{
F_\P \, \widetilde{\phi}_! \, F_\O (A) \ar@<0.5ex>[r]^-{\partial_0}  \ar@<-0.5ex>[r]_-{\partial_1} ~&~ 
F_\P \, \widetilde{\phi}_!(A)
}\Big)\quad
\end{flalign}
in $\Alg_\P(\MM)$, where $\partial_0$ and $\partial_1$ are given in \eqref{eqn:partial0} and \eqref{eqn:partial1},
and the reflector is \eqref{eqn:reflectorrrrr}.
\end{propo}

\begin{rem}\label{rem:coloradjunctioncomposition}
Forming the pullback functor according to \eqref{eqn:fphiastAlg} preserves identities and 
compositions, i.e.\ $\id_\O^\ast = \id_{\Alg_\O(\MM)}$ 
and $(\psi \, \phi)^\ast = \phi^\ast \, \psi^\ast$, 
for all $\O\in\Op$ and all $\Op$-morphisms 
$\phi: \O \to \P$ and $\psi: \P \to \mathcal{Q}$.
Hence, from uniqueness (up to a unique natural isomorphism) 
of adjoint functors, we obtain also natural isomorphisms 
${\id_\O}_! \cong \id_{\Alg_\O(\MM)}$ 
and $(\psi \, \phi)_! \cong \psi_! \, \phi_!$. 
\end{rem}

\begin{rem}\label{rem:rightadjointofchangeofcolorpreservesreflexivecoeq}
We record for later use that the pullback functor
$\phi^\ast: \Alg_\P(\MM) \to \Alg_\O(\MM)$ 
along any $\Op$-morphism $\phi: \O \to \P$
preserves reflexive coequalizers:
Let us denote by $C\in \Alg_\P(\MM)$ the coequalizer for a reflexive pair of 
parallel $\Alg_\P(\MM)$-morphisms $A \rightrightarrows B$.
By Proposition \ref{propo:colimitsinAlg},
$\phi^\ast(C)$ is a reflexive coequalizer 
for $\phi^\ast(A) \rightrightarrows \phi^\ast(B)$ 
in $\Alg_\O(\MM)$ if and only if 
$U_\O \, \phi^\ast(C)$ is a reflexive coequalizer for $U_\O \, \phi^\ast(A) \rightrightarrows 
U_\O \, \phi^\ast(B)$ in $\MM^\CCC$. Using also
$U_\O \, \phi^\ast = \widetilde{\phi}^\ast \, U_\P$ and that $U_\P$ 
preserves reflexive coequalizers (cf.\ Proposition \ref{propo:colimitsinAlg}),
our claim follows from the fact that $\widetilde{\phi}^\ast: \MM^\DDD \to \MM^\CCC$ 
preserves colimits because it has a right adjoint given by {\em right} Kan extension
along the underlying map of sets of colors $\widetilde{\phi}: \CCC\to\DDD$.
\end{rem}

%%%%%%%%%%%%%%%%%%%%%%%%%%%%%%%%%%%%%%%%%%%%%%%%
%%%%%%%%%%%%%%%%%%%%%%%%%%%%%%%%%%%%%%%%%%%%%%%%

\section{\label{sec:QFToperads}Quantum field theory operads}
The aim of this section is to construct colored operads whose algebras describe quantum field theories.
Abstractly, we assign to any small category $\CC$ with orthogonality relation $\perp$
(cf.\ Definition \ref{def:orthcat} below) a colored 
operad that is inspired by the structures underlying algebraic quantum field theory.
The category $\CC$ should be interpreted as the category of ``spacetimes'' of interest, 
and the orthogonality relation $\perp$ encodes the commutative behavior of certain observables.
Our construction is very flexible and in particular it reveals an operadic structure underlying various
kinds of quantum field theories, including Haag-Kastler theories on Minkowski spacetime 
\cite{HaagKastler}, locally covariant theories on all Lorentzian spacetimes \cite{Brunetti,FewsterVerch},
chiral conformal theories \cite{Kawahigashi,Rehren,BDHcft} and also Euclidean theories \cite{Schlingemann}.
Each of these models is obtained by a different choice of category $\CC$ and orthogonality relation $\perp$, 
however the operadic structure is formally the same.

\subsection{\label{subsec:motivationQFT}Motivation}
Let $\MM$ be a bicomplete closed symmetric monoidal category
and $\CC$ a small category, which we interpret as the category of spacetimes of interest.
The basic idea of algebraic quantum field theory is to assign coherently to each spacetime $c\in\CC$
an associative and unital algebra $\mathfrak{A}(c) \in\Alg(\MM)$ in $\MM$ (cf.\ Example \ref{ex:Alg}), 
i.e.\ we consider a functor
\begin{flalign}\label{eqn:funCCAlg}
\mathfrak{A}\,:\,\CC~\longrightarrow~\Alg(\MM)\quad.
\end{flalign}
The algebra $\mathfrak{A}(c)$ is interpreted as the ``algebra of quantum observables'' 
that can be measured in the spacetime $c\in \CC$. The  algebra morphism
$\mathfrak{A}(f) : \mathfrak{A}(c)\to\mathfrak{A}(c^\prime)$ is the pushforward
of observables along the $\CC$-morphism $f:c\to c^\prime$ which is
interpreted as a spacetime embedding. 
\begin{rem}\label{rem:involutions}
For applications to physics, it is often necessary to consider associative and unital 
{\em $\ast$-algebras} over the complex numbers $\bbC$. We however decided to ignore
in the present paper the concept of $\ast$-involutions, because this makes our
presentation of the operadic structure underlying algebraic quantum field theories
more transparent. 
Using the technically more involved concept of {\em involutive categories} 
as target categories, one can generalize the results of the present paper 
to the case including $\ast$-involutions. 
We refer to our follow-up paper \cite{BSWinvolutions} where this has been
carried out in detail.
\end{rem}

It is crucial to demand that the functor 
$\mathfrak{A}$ satisfies a collection of physically motivated axioms. Inspired by 
the time-slice axiom in Lorentzian theories \cite{HaagKastler,Brunetti,FewsterVerch},
we propose the following formalization and generalization.
\begin{defi}\label{def:Wconstant}
Let $W\subseteq \mathrm{Mor}\,\CC$ be a subset of the set of  morphisms in $\CC$.
A functor $\mathfrak{A} : \CC\to\Alg(\MM)$ is called {\em $W$-constant}
if it sends every morphism $f :c\to c^\prime$ in $W$ to an isomorphism 
$\mathfrak{A}(f) : \mathfrak{A}(c) \stackrel{\cong}{\longrightarrow} \mathfrak{A}(c^\prime)$ in $\Alg(\MM)$.
\end{defi}
\begin{rem}
Our notion of $W$-constancy is similar to and generalizes the concept of {\em local constancy}
of factorization algebras \cite{CostelloGwilliam}. This motivates our choice of terminology.
\end{rem}
By \cite[Proposition~5.2.2]{handbook1}, the localization $\CC[W^{-1}]$ 
of the small category $\CC$ at $W$ exists as a small category.
Hence, we may equivalently describe $W$-constant 
functors as functors  $\mathfrak{A} : \CC[W^{-1}]\to \Alg(\MM)$
on the localized category. This means that the $W$-constancy axiom may 
be implemented formally by choosing $\CC[W^{-1}]$ instead of $\CC$ 
as the underlying category. 
\sk

Another important property of quantum field theories is that certain pairs of observables
behave commutatively. For example, in Lorentzian theories \cite{HaagKastler,Brunetti,FewsterVerch}
any two observables which are localized in spacelike separated regions commute with each other.
Similarly, in Euclidean \cite{Schlingemann} and chiral conformal \cite{Kawahigashi,Rehren,BDHcft} theories
any two observables localized in disjoint regions commute. We propose the following formalization 
and generalization of this concept:
Let us denote by $\mathrm{Mor}\,\CC \, {}_\mathrm{t}^{}\!\times\!{}^{}_\mathrm{t} \,\mathrm{Mor}\,\CC$
the set of pairs of $\CC$-morphisms whose targets coincide. Its elements $(f_1,f_2)$ may be visualized
as $\CC$-diagrams of the form $c_1 \stackrel{f_1}{\longrightarrow} c \stackrel{f_2}{\longleftarrow} c_2 $.
\begin{defi}\label{def:orthcat}
Let $\CC$ be a small category.
An {\em orthogonality relation} $\perp$ on $\CC$ is a subset 
${\perp} \subseteq \mathrm{Mor}\,\CC \, {}_\mathrm{t}^{}\!\times\!{}^{}_\mathrm{t} \,\mathrm{Mor}\,\CC$
satisfying the following properties:
\begin{itemize}
\item[(1)] {\em Symmetry:} $(f_1,f_2)\in {\perp} ~\Longrightarrow~(f_2,f_1)\in{\perp}$.
\item[(2)] {\em Stability under post-composition:} $(f_1,f_2)\in{\perp} ~\Longrightarrow~
(g\,f_1, g\, f_2)\in{\perp}$, for all composable $\CC$-morphisms $g$.
\item[(3)] {\em Stability under pre-composition:} $(f_1,f_2)\in{\perp} ~\Longrightarrow~
(f_1\,h_1,  f_2\,h_2)\in{\perp}$, for all composable $\CC$-morphisms $h_1$ and $h_2$.
\end{itemize}
We call elements $(f_1,f_2) \in {\perp}$ {\em orthogonal pairs} and write also $f_1 \perp f_2$.
A pair $\ovr{\CC}:= (\CC,\perp)$ consisting of a small category and an orthogonality relation
is called an {\em orthogonal category}. 
A morphism $F : \ovr{\CC}=(\CC,\perp_\CC)\to \ovr{\DD} = (\DD,\perp_\DD)$ 
of orthogonal categories is a functor $F : \CC\to\DD$ 
that preserves the orthogonality relations, 
i.e.\ $f_1 \perp_\CC f_2 ~\Longrightarrow~ F(f_1) \perp_\DD F(f_2)$. 
We call such morphisms {\em orthogonal functors} and denote by 
$\OCat$ the {\em category of orthogonal categories}. 
\end{defi}

\begin{defi}\label{def:perpcommutative}
Let $\ovr{\CC}=(\CC,\perp)$ be an orthogonal category.
\begin{itemize}
\item[a)] A functor $\mathfrak{A}: \CC \to \Alg(\MM)$ is called {\em $\perp$-commutative over $c\in\CC$} 
if for every $f_1\perp f_2$ with target $c$ the diagram
\begin{flalign}\label{eqn:perpcomdiagram}
\xymatrix{
\ar[d]_-{\mathfrak{A}(f_1)\otimes\mathfrak{A}(f_2)}\mathfrak{A}(c_1) \otimes \mathfrak{A}(c_2) \ar[rr]^-{\mathfrak{A}(f_1)\otimes\mathfrak{A}(f_2)} ~&& ~
\mathfrak{A}(c)\otimes \mathfrak{A}(c) \ar[d]^-{\mu_c^\op}\\
\mathfrak{A}(c)\otimes \mathfrak{A}(c) \ar[rr]_-{\mu_c} ~&&~\mathfrak{A}(c)
}
\end{flalign}
in $\MM$ commutes. Here $\mu_c$ denotes the 
multiplication on $\mathfrak{A}(c)$ and $\mu_c^\op:= \mu_c\, \tau$
the opposite multiplication on $\mathfrak{A}(c)$, 
with $\tau$ the symmetric braiding of $\MM$.

\item[b)] A functor $\mathfrak{A}: \CC \to \Alg(\MM)$ 
is called {\em $\perp$-commutative} if it is
$\perp$-commutative over every $c\in\CC$.

\item[c)] The full subcategory of the functor category $\Alg(\MM)^\CC$ 
whose objects are all $\perp$-commutative functors 
is denoted by $\Alg(\MM)^{\ovr{\CC}}$. 
There exists a fully faithful forgetful functor
\begin{flalign}
U\,:\, \Alg(\MM)^{\ovr{\CC}} ~\longrightarrow~ \Alg(\MM)^{\CC} \quad,
\end{flalign}
which forgets $\perp$-commutativity.
\end{itemize}
\end{defi}

Notice that, in contrast to $W$-constancy, $\perp$-commutativity can not be implemented easily
by adjusting the categories $\CC$ and $\Alg(\MM)$: It imposes a non-trivial condition on
functors $\mathfrak{A} : \CC \to \Alg(\MM)$ that relates orthogonal pairs of morphisms 
$f_1\perp f_2$ in $\CC$ to a certain commutative behavior of the algebraic structures 
on the objects of $\MM$ underlying $\mathfrak{A}$.
The colored operads we develop in this section solve this problem in the sense that they allow us
to encode $\perp$-commutativity as part of the structure 
instead of enforcing it as a property. This is very useful
for the study of universal constructions for quantum field theories, see Sections
\ref{sec:QFTadjunctions} and \ref{sec:QFTconstructions}.

\subsection{\label{subsec:SetOperads}Definition and properties}
We assign to each orthogonal category 
$\ovr{\CC}=(\CC,\perp)$ (cf.\ Definition \ref{def:orthcat}) 
a colored operad $\O_{\ovr{\CC}} \in \Op$ inspired 
by the algebraic structures underlying quantum field theories. 
For this we first construct an auxiliary 
colored operad $\O_\CC \in \Op$, 
neglecting the orthogonality relation, 
and then define $\O_{\ovr{\CC}} \in \Op$
in terms of an equivalence relation determined by $\perp$. 
The auxiliary colored operad $\O_\CC$ governs $\CC$-diagrams 
of associative and unital algebras, see Remark \ref{rem:AlgoverOCC} below. 
For its definition, recall the associative operad $\As\in\Op$ from Example \ref{ex:Op}.
\begin{defi}\label{def:OCC}
Let $\CC$ be a small category. The colored operad $\O_\CC \in \Op$ 
consists of the following data: 
\begin{itemize}
\item[a)] the set of colors $\CC_0$ is the set of objects of $\CC$, 

\item[b)] for all $t \in \CC_0$, $n \geq 0$, $\und{c} \in \CC_0^n$, 
the set of operations is $\O_\CC\big(\substack{t\\\und{c}}\big) 
:= \As(n) \times \CC(\und{c},t)$, where $ \CC(\und{c},t) := \prod_{i=1}^n \CC(c_i,t)$ and $\CC(c_i,t) 
:= \Hom_{\CC}(c_i,t)$ is the set of morphisms from $c_i$ to $t$,
whose elements are denoted by $(\sigma,\und{f}) := (\sigma,(f_1,\dots,f_n))\in \O_\CC \big(\substack{t\\\und{c}}\big) $,

\item[c)] for all $t \in \CC_0$, $n \geq 0$, $\und{c} \in \CC_0^n$, 
$\sigma^\prime \in \Sigma_n$, the map 
\begin{flalign}
\O_\CC(\sigma^\prime) \,:\, \O_\CC\big(\substack{t\\\und{c}}\big)~&\longrightarrow~ \O_\CC\big(\substack{t\\\und{c}\sigma^\prime}\big)~~,\quad
(\sigma,\und{f})~\longmapsto~\big(\sigma\sigma^\prime,\und{f}\sigma^\prime\big)\quad,\label{eqn:OCCperm}
\end{flalign}
where we use the permutation action of $\As$ and define $\und{f}\sigma^\prime := (f_{\sigma^\prime(1)},\dots,f_{\sigma^\prime(n)})$, 

\item[d)] for all $t \in \CC_0$, $m \geq 1$, $\und{a} \in \CC_0^m$, 
$k_i \geq 0$, $\und{b}_i \in \CC_0^{k_i}$, for $i = 1, \dots, m$, 
the composition map 
\begin{flalign}
\nn \gamma \,:\, \O_\CC\big(\substack{t\\\und{a}}\big) \times \prod_{i=1}^m \O_\CC\big(\substack{a_i\\\und{b}_i}\big) ~&\longrightarrow~ \O_\CC\big(\substack{t\\(\und{b}_1,\dots,\und{b}_m)}\big) \quad, \\
\big( (\sigma,\und{f}), \big((\sigma_1,\und{g}_1),\dots,(\sigma_m,\und{g}_m)\big) \big) ~&\longmapsto~
\big(\sigma(\sigma_1, \dots, \sigma_m) , \und{f}(\und{g}_1,\dots,\und{g}_m)\big)\quad, \label{eqn:OCCcomp}
\end{flalign}
where we use the composition of $\As$ and define
\begin{flalign}\label{eqn:CCcomp}
\und{f}(\und{g}_1,\dots,\und{g}_m) \,:=\, \big(f_1\,g_{11},\dots, f_{1}\,g_{1k_1},\dots, f_{m}\,g_{m 1},\dots, f_{m}\,g_{m k_m}\big) 
\end{flalign}
by composition in the category $\CC$,

\item[e)] for all $t\in\CC_0$, the unit 
$(e, \id_t) \in \O_\CC\big(\substack{t\\t}\big)$, 
where we use the unit of $\As$ and the identity 
morphism $\id_t \in \CC(t,t)$. 
\end{itemize}
\end{defi}
\begin{rem}
The operad $\O_\CC$ has a natural physical interpretation
coming from algebraic quantum field theory.
We may graphically visualize an element 
$(\sigma,\und{f})\in \O_\CC\big(\substack{t\\\und{c}}\big)$ by
\begin{flalign}\label{eqn:basicoperationpic}
\begin{tikzpicture}[cir/.style={circle,draw=black,fill=black,inner sep=0pt,minimum size=2mm},
        poin/.style={circle, inner sep=0pt,minimum size=0mm}]
%
% Vertex position/input/output
%
  \node[poin] (V) at (0,0) {};
  \node[poin] (outV) [label=above:{\small $t$}] at (0,0.75) {}; 
  \node[poin] (inVl) at (-0.75,-0.75) {};
  \node[poin] (inVr)  at (0.75,-0.75) {}; 
  \node[poin] (dotsinV) at (0,-0.5) {$\cdots$};
%
% Draw lines to inputs and output
%
\draw[thick] (V) -- (outV);
\draw[thick] (inVl) -- (V);
\draw[thick] (inVr) -- (V);
%\draw[fill=black] (V) circle (0.6mm);
%
%  Horizontal vertex profile line
%
  \node[poin] (PVl) at (-1.25,-0.75) {};
  \node[poin] (PVm) at (0,-0.85) {}; 
  \node[poin] (PVr)  [label=right:{\small $t^{n}\sigma^{-1}$}] at (1.25,-0.75) {};
%
% Draw horizontal vertex profile line
%
\draw[thick] (PVl) -- (PVr);
%
% Next horizontal line
%
  \node[poin] (mapl) at (-1.25,-1.5) {};
  \node[poin] (mapm) at (0,-1.4) {}; 
  \node[poin] (mapr) [label=right:{\small $t^{n}$}]  at (1.25,-1.5) {};
  \node[poin] (mapml) at (-0.75,-1.5) {};
  \node[poin] (mapmr) at (0.75,-1.5) {}; 
  \node[poin] (dotsmap) at (0,-1.875) {$\cdots$};  
%
% Draw next horizontal line and permutation
%
\draw[thick] (mapl) -- (mapr);
\draw[thick, ->] (PVm) -- (mapm) node[midway,left] {{\small $\sigma$}};
%
% Input profile line
%
  \node[poin] (inl) at (-1.25,-2.25) {};
  \node[poin] (inm) at (0,-2.15) {}; 
  \node[poin] (inr) [label=right:{\small $\und{c}$}] at (1.25,-2.25) {};
  \node[poin] (inml) at (-0.75,-2.25) {};
  \node[poin] (inmr) at (0.75,-2.25) {}; 
%
% Draw input horizontal line and f's
%
\draw[thick] (inl) -- (inr);
\draw[thick] (inml) -- (mapml) node[midway,left] {{\small $f_1$}};
\draw[thick] (inmr) -- (mapmr) node[midway,right] {{\small $f_n$}};
\end{tikzpicture}
\end{flalign}
This picture should be read from bottom to top and
it should be understood as the following 
$3$-step operation in algebraic quantum field theory.
(1)~Apply the morphisms $\und{f}$ to observables on $\und{c}=(c_1,\dots,c_n)$; 
(2)~Permute the resulting observables on $t^{n}$
by acting with $\sigma^{-1}$ from the right;
(3)~Multiply the resulting observables on $t^{n}\sigma^{-1}$ 
according to the order in which they appear.
Concretely, given observables $\Phi_i$ on $c_i$, for 
$i=1,\dots,n$, the $3$-step operation formally looks like
\begin{flalign}
\nn\Phi_1\otimes \cdots \otimes \Phi_n ~~&\stackrel{\text{(1)}}{\longmapsto}~~f_1(\Phi_1) \otimes \cdots\otimes f_n(\Phi_n) \\
\nn~~&\stackrel{\text{(2)}}{\longmapsto}~~f_{\sigma^{-1}(1)}(\Phi_{\sigma^{-1}(1)})\otimes \cdots\otimes 
f_{\sigma^{-1}(n)}(\Phi_{\sigma^{-1}(n)})\\
~~&\stackrel{\text{(3)}}{\longmapsto}~~ f_{\sigma^{-1}(1)}(\Phi_{\sigma^{-1}(1)})\, \cdots\,
f_{\sigma^{-1}(n)}(\Phi_{\sigma^{-1}(n)})\quad.\label{eqn:operationordering}
\end{flalign}
The operad structure on $\O_\CC$ presented in Definition \ref{def:OCC}
captures how such $3$-step operations 
in algebraic quantum field theory compose. 
\end{rem}

The colored operad $\O_\CC$ from Definition \ref{def:OCC} 
does not encode the orthogonality relation $\perp$ 
on the category $\CC$ yet. We impose relations on $\O_{\CC}$
that are inspired by the $\perp$-commutativity axiom 
(cf.\ Definition \ref{def:perpcommutative}) 
and our quantum field theoretic interpretation
given in \eqref{eqn:operationordering}.
\begin{defi}\label{defi:Rperp}
Let $\ovr{\CC}=(\CC,\perp)$ be an orthogonal category.
We define a $\CC_0$-colored sequence $R_\perp \in \Seq_{\CC_0}$ 
by setting, for all $t \in \CC_0$, $n \geq 0$, $\und{c} \in \CC_0^n$, 
\begin{flalign}\label{eqn:Rperp}
R_\perp\big( \substack{t \\ \und{c}} \big) ~:=~ 
\begin{cases}
{\perp} \cap \CC (\und{c},t) &~,~~ \text{if } n = 2 \quad, \\
\emptyset &~,~~ \text{else} \quad.
\end{cases}
\end{flalign}
We further define a pair of parallel $\Seq_{\CC_0}$-morphisms
$s_{\perp,1}, s_{\perp,2}: R_\perp\rightrightarrows \O_{\CC}$ by setting 
\begin{subequations}\label{eqn:rperp}
\begin{flalign}
s_{\perp,1}\,:\, R_\perp \big( \substack{t \\ \und{c}} \big) 
& ~\longrightarrow~ \O_{\CC} \big( \substack{t \\ \und{c}} \big) ~~,\quad 
(f_1,f_2) ~\longmapsto~ \big( e, (f_1,f_2) \big) \quad, \\
s_{\perp,2}\,:\, R_\perp\big( \substack{t \\ \und{c}} \big) 
& ~\longrightarrow~ \O_{\CC} \big( \substack{t \\ \und{c}} \big) ~~, \quad
(f_1,f_2) ~\longmapsto~ \big( \tau, (f_1,f_2) \big) \quad, 
\end{flalign}
\end{subequations}
for all $t \in \CC_0$ and $\und{c} \in \CC_0^2$,
where $e, \tau \in \Sigma_2$ are the group identity and the transposition. 
The $\CC_0$-colored operad $\O_{\ovr\CC} \in \Op_{\CC_0}$ 
is defined as the coequalizer 
\begin{flalign}
\O_{\ovr{\CC}} ~:=~ \colim \Big( 
\xymatrix{
F(R_\perp) \ar@<0.5ex>[r]^-{s_{\perp,1}} 
\ar@<-0.5ex>[r]_-{s_{\perp,2}} ~&~ \O_{\CC}
}
\Big) 
\end{flalign}
in $\Op_{\CC_0}$, where $F$ denotes the free $\CC_0$-colored operad functor, cf.\
Theorem \ref{theo:freeoperad}. By construction, there exists a 
canonical $\Op_{\CC_0}$-morphism $p_{\ovr{\CC}}: \O_\CC\to \O_{\ovr{\CC}}$. 
\end{defi}

Note that $\O_{\ovr{\CC}} \in \Op_{\CC_0}$ exists 
because $\Op_{\CC_0}$ is cocomplete (cf.\ Proposition \ref{propo:colimitsinOp}), 
however its definition as a coequalizer is quite implicit. 
The next proposition establishes a fully explicit presentation 
of $\O_{\ovr{\CC}}$ in terms of an equivalence relation 
on the sets of operations of $\O_{\CC}$. 
\begin{propo}\label{propo:OCCbarexplicit}
Let $\ovr{\CC}=(\CC,\perp)$ be an orthogonal category.
For all $t \in \CC_0$, $n \geq 0$, $\und{c} \in \CC_0^n$, 
equip the set $\O_\CC\big(\substack{t\\\und{c}}\big)$
from Definition \ref{def:OCC} with the following equivalence relation: 
$(\sigma,\und{f}) \sim_\perp (\sigma^\prime,\und{f^\prime})$ if and only if
\begin{enumerate}
\item[(1)] $\und{f} = \und{f}^\prime$ as elements in $\CC(\und{c},t)$,
\item[(2)] the right permutation $\sigma\sigma^{\prime -1} : \und{f}\sigma^{-1}\to \und{f}\sigma^{\prime -1}$
is generated by transpositions of adjacent orthogonal pairs, i.e.\
$\sigma\sigma^{\prime -1} = \tau_1\cdots\tau_N$, where $N\in\bbN$
is a positive integer and $\tau_1,\dots,\tau_N \in\Sigma_{n}$
are transpositions such that, for all $k=1,\dots,N$, the right permutation
\begin{flalign}
\tau_k \,: \, \und{f}\sigma^{-1}\tau_1\cdots\tau_{k-1} ~\longrightarrow~ \und{f}\sigma^{-1}\tau_1\cdots\tau_k
\end{flalign}
is a transposition of a pair of $\CC$-morphisms that 
are orthogonal and adjacent in the sequence $\und{f}\sigma^{-1}\tau_1\cdots\tau_{k-1}$.
\end{enumerate}
The colored operad structure on $\O_\CC\in\Op_{\CC_0}$ descends to the quotient sets
$\O_\CC\big(\substack{t\\\und{c}}\big)/{\sim_\perp}$,
which defines a model $\O_{\ovr{\CC}} := \O_\CC/{\sim_\perp}\in\Op_{\CC_0}$ 
for the coequalizer in Definition \ref{defi:Rperp}. The canonical
$\Op_{\CC_0}$-morphism $p_{\ovr{\CC}}: \O_\CC \to \O_{\ovr{\CC}}$ 
is given by the quotient map $\pi : \O_\CC \to \O_\CC/{\sim_\perp}$. 
\end{propo}
\begin{ex}
Before we prove Proposition \ref{propo:OCCbarexplicit}, 
let us illustrate the equivalence relation $\sim_\perp$ with a simple example. 
For $t\in\CC_0$ and $\und{c} = (c_1,c_2) \in \CC_0^2$,
elements in $\O_\CC\big(\substack{t\\\und{c}}\big)$ 
are either of the form $(e,(f_1,f_2))$ or of the form $(\tau,(f_1,f_2))$,
where $e\in\Sigma_2$ denotes the identity permutation and $\tau\in\Sigma_2$
the transposition. It follows that $(e, (f_1,f_2)) \sim_\perp (\tau,(f_1,f_2))$ if and only if $f_1 \perp f_2$. 
Indeed, using $\tau^{-1} = \tau \in \Sigma_2$,
the right permutation 
\begin{flalign}
e \tau^{-1} = \tau ~ : ~ (f_1,f_2) e^{-1} = (f_1,f_2) ~~\longrightarrow ~~  (f_1,f_2) \tau^{-1} = (f_2,f_1) 
\end{flalign}
is a transposition of an adjacent orthogonal pair if and only if  $f_1\perp f_2$. 
\end{ex}
\begin{proof}[Proof of Proposition \ref{propo:OCCbarexplicit}]
Let us first show that the colored operad structure on $\O_\CC$ (cf.\ Definition \ref{def:OCC}) descends to 
the quotient $\O_\CC/{\sim_\perp}$. The parts concerning the permutation action 
and unit are straightforward, while composition requires a calculation 
using some basic properties of permutations. 
Let $t\in\CC_0$, $m \geq 1$, $\und{a} \in \CC_0^m$, 
$k_i \geq 0$, $\und{b}_i \in \CC_0^{k_i}$, for $i=1,\dots,m$, 
and consider $(\sigma, \und{f}) \sim_\perp (\widetilde\sigma, \und{f})$ 
in $\O_\CC\big(\substack{t\\\und{a}}\big)$, 
$(\sigma_i, \und{g}_i) \in \O_\CC\big(\substack{a_i\\\und{b}_i}\big)$, 
for $i=1,\dots,m$. We have to show that 
\begin{subequations}\label{eqn:comprelation1}
\begin{flalign}
\gamma \Big( (\sigma,\und{f}), \big((\sigma_1,\und{g}_1),\dots, (\sigma_m,\und{g}_m)\big) \Big) 
& = \Big( \sigma(\sigma_1,\dots,\sigma_m), \und{f}(\und{g}_{1},\dots,\und{g}_m) \Big) \quad, \\
\gamma \Big( (\widetilde\sigma,\und{f}), \big((\sigma_1,\und{g}_1),\dots, (\sigma_m,\und{g}_m)\big)\Big) 
& = \Big( \widetilde\sigma(\sigma_1,\dots,\sigma_m), \und{f}(\und{g}_{1},\dots,\und{g}_m) \Big) 
\end{flalign}
\end{subequations}
are equivalent under $\sim_{\perp}$ in $\O_\CC\big(\substack{t\\ (\und{b}_1,\dots,\und{b}_m) }\big)$.
Thus, the right permutation to be investigated is 
\begin{flalign}
\xymatrix{
\und{f} ( \und{g}_{1}, \dots, \und{g}_m ) ~ \big(\sigma (\sigma_1, \dots, \sigma_m)\big)^{-1} 
\ar[d]^-{~\sigma (\sigma_1, \dots, \sigma_m) ~ (\widetilde\sigma (\sigma_1, \dots, \sigma_m))^{-1}}\\
\und{f} ( \und{g}_{1}, \dots, \und{g}_m ) ~ \big(\widetilde\sigma (\sigma_1, \dots, \sigma_m)\big)^{-1} 
}
\end{flalign}
Recalling that $\sigma (\sigma_1, \dots, \sigma_m)$ is defined 
as the group multiplication of the block permutation 
$\sigma\langle k_{\sigma^{-1}(1)},\dots, k_{\sigma^{-1}(m)} \rangle$ 
and the block sum permutation $\sigma_1\oplus\cdots\oplus\sigma_m$
(cf.\ Example \ref{ex:Op}), we find the equivalent expression 
\begin{flalign}
\xymatrix{
(\und{f} \sigma^{-1}) \big( \und{g}_{\sigma^{-1}(1)} \sigma_{\sigma^{-1}(1)}^{-1}, \dots, \und{g}_{\sigma^{-1}(m)} \sigma_{\sigma^{-1}(m)}^{-1} \big)
\ar[d]^-{~( \sigma \widetilde\sigma^{-1} ) \langle k_{\sigma^{-1}(1)}, \dots, k_{\sigma^{-1}(m)} \rangle}\\
(\und{f} \widetilde\sigma^{-1}) \big( \und{g}_{\widetilde\sigma^{-1}(1)} \sigma_{\widetilde\sigma^{-1}(1)}^{-1}, \dots, \und{g}_{\widetilde\sigma^{-1}(m)} \sigma_{\widetilde\sigma^{-1}(m)}^{-1} \big)
}
\end{flalign}
for this right permutation.
From the latter expression we deduce that this right permutation is generated
by transpositions of adjacent orthogonal pairs because
(1)~$\sigma \widetilde\sigma^{-1}: \und{f} \sigma^{-1} \to \und{f} \widetilde\sigma^{-1}$ 
is by hypothesis generated by transpositions of adjacent orthogonal pairs,
(2)~$\perp$ is by definition stable under pre-composition in $\CC$, and 
(3)~transpositions of adjacent blocks are generated by adjacent transpositions 
of block elements belonging to different blocks. By a similar argument, one shows that
the composition of $(\sigma, \und{f})$ with equivalent 
$(\sigma_i, \und{g}_i)\sim_\perp (\widetilde{\sigma}_i,\und{g}_i)$, for
$i=1,\dots,m$, gives $\sim_\perp$-equivalent results. Hence, the composition map
$\gamma$ descends to the quotient $\O_\CC/{\sim_\perp}$.
\sk

It remains to prove that the quotient map
$\pi: \O_\CC\to \O_{\CC}/{\sim_\perp}$ is the coequalizer of 
the pair of parallel morphisms $s_{\perp,1}, s_{\perp,2}$ from Definition \ref{defi:Rperp}. 
It directly follows from the definition of $\sim_\perp$ 
that  $\pi\, s_{\perp,1} = \pi\, s_{\perp,2}$.
To show that $\pi$ has also the required universal property, 
suppose that $\phi: \O_\CC\to \P$ in $\Op_{\CC_0}$ 
is such that $\phi\, s_{\perp,1} = \phi\, s_{\perp,2}$. We have to
prove that $\phi$ factors uniquely through $\pi$. 
Uniqueness of the factorization follows from surjectivity of the maps 
$\pi: \O_\CC\big(\substack{t\\\und{c}}\big) \to 
\O_\CC\big(\substack{t\\\und{c}}\big)/{\sim_\perp}$, 
therefore it only remains to show that 
$\phi(\sigma,\und{f}) = \phi(\sigma^\prime,\und{f}) \in \P\big(\substack{t\\\und{c}}\big)$, for all  
$(\sigma,\und{f}) \sim_\perp (\sigma^\prime,\und{f}) \in 
\O_\CC\big(\substack{t\\\und{c}}\big)$. 
Equivariance under permutations leads to the equivalent condition
$\phi(\sigma\sigma^{\prime -1},\und{f}\sigma^{\prime -1}) 
= \phi(e,\und{f}\sigma^{\prime -1})$. By definition of $\sim_\perp$, 
we find a factorization $\sigma\sigma^{\prime -1} =  
\tau_1\cdots\tau_N: \und{f}\sigma^{-1} \to \und{f}\sigma^{\prime -1}$ 
into transpositions of adjacent orthogonal pairs. 
Let us focus first on the inverse of the rightmost transposition, i.e.\ 
$\tau_N^{-1}: \und{f}\sigma^{\prime -1 } \to \und{f}\sigma^{\prime -1}\tau_N^{-1}$. 
It transposes a pair of adjacent orthogonal morphisms, say 
$(f_{\sigma^{\prime-1}(k)},f_{\sigma^{\prime-1}(k+1)})\in {\perp}$,
in the sequence $\und{f}\sigma^{\prime -1 }$. 
Using the composition $\gamma$ of $\O_\CC$ given in \eqref{eqn:OCCcomp},
we may decompose $(e,\und{f}\sigma^{\prime -1})$ as
\begin{flalign}\label{eqn:operationfactorization}
\resizebox{.9\hsize}{!}{$\gamma\Big( \big(e,\underbrace{(\id_t,\dots,\id_t)}_{\text{$n{-}1$ times}}\big), \Big(
\big(e,f_{\sigma^{\prime -1}(i)}\big)_{i=1,\dots,k-1}, 
\big(e,(f_{\sigma^{\prime-1}(k)},f_{\sigma^{\prime-1}(k+1)})\big), 
\big(e,f_{\sigma^{\prime -1}(i)}\big)_{i=k+2,\dots,n}\Big) \Big)\quad.$}
\end{flalign}
Applying $\phi$ to this expression, we observe that the term 
$\phi\big(e,(f_{\sigma^{\prime-1}(k)},f_{\sigma^{\prime-1}(k+1)})\big)$
can be replaced by  
$\phi\big(\tau,(f_{\sigma^{\prime-1}(k)},f_{\sigma^{\prime-1}(k+1)})\big)$
since $\phi\, s_{\perp,1} = \phi\, s_{\perp,2}$, which proves the identity
\begin{flalign}
\phi\big(e,\und{f}\sigma^{\prime -1}\big) = \phi \big(\tau_N,\und{f}\sigma^{\prime -1}\big)
= \P(\tau_N) \phi \big(e,\und{f}\sigma^{\prime -1}\tau_N^{-1}\big)\quad.
\end{flalign}
Iterating this construction then proves the desired equality 
\begin{flalign}
\phi\big(e,\und{f}\sigma^{\prime -1}\big) = \P(\tau_N)\cdots\P(\tau_1) 
\phi \big(e,\und{f}\sigma^{\prime -1}\tau_N^{-1}\cdots\tau_1^{-1}\big)
= \phi\big(\sigma\sigma^{\prime -1},\und{f}\sigma^{\prime -1}\big)
\end{flalign}
and completes the proof.
\end{proof}

We conclude this subsection with an important remark on 
the functoriality of the constructions in Definitions \ref{def:OCC} and \ref{defi:Rperp}. 
The proof of the following result is straightforward.
\begin{propo}\label{propo:functoriality}
The assignment $\CC \mapsto \O_\CC$  
naturally extends to a functor $\O_{(-)}: \Cat \to \Op$ 
on the category of small categories. 
Explicitly, to each functor $F :\CC\to\DD$ we assign 
the $\Op$-morphism $\O_F : \O_\CC \to \O_\DD$
consisting of the following data: 
\begin{itemize}
\item[a)] the map of colors $\widetilde{\O}_F := F : \CC_0 \to \DD_0$
is given by the action of the functor $F $ on objects,

\item[b)] for all $t \in \CC_0$, $n \geq 0$, $\und{c} \in \CC_0^n$, 
the map $\O_F : \O_\CC\big(\substack{t\\\und{c}}\big) 
\to \O_\DD\big(\substack{F(t)\\F(\und{c})}\big)\,,~
(\sigma,\und{f}) \mapsto \big(\sigma,F(\und{f})\big)$
is given by the action of the functor $F$ on morphisms, i.e.\
$F(\und{f}) := (F(f_1),\dots,F(f_n))$.
\end{itemize}
Similarly, the assignment $\ovr{\CC} \mapsto \O_{\ovr{\CC}} := \O_\CC/{\sim_\perp}$ 
naturally extends to a functor $\O_{(-)} : \OCat \to \Op$ on the category of orthogonal 
categories. Denoting by $\Pi : \OCat \to \Cat \,,~(\CC,\perp)\mapsto \CC$ the evident forgetful functor,
the $\Op$-morphisms $p_{\ovr{\CC}}: \O_{\CC} \to \O_{\ovr{\CC}}$ 
are the components of a natural transformation $p : \O_{\Pi(-)}  \to \O_{(-)}$
between functors from $\OCat$ to $\Op$. 
\end{propo}

\begin{rem}\label{rem:pCCviafunctoriality}
Notice that for the empty orthogonality relation $\emptyset \subseteq 
\mathrm{Mor}\,\CC \, {}_\mathrm{t}^{}\!\times\!{}^{}_\mathrm{t} \,\mathrm{Mor}\,\CC$
the colored operad $\O_{(\CC,\emptyset)}$ coincides with $\O_{\CC}$.
Given any $\ovr{\CC} = (\CC,\perp) \in\OCat$, 
the identity functor $\id_\CC : (\CC,\emptyset)\to\ovr{\CC}$ is
clearly orthogonal, hence we obtain by Proposition \ref{propo:functoriality}
an $\Op_{\CC_0}$-morphisms $\O_{\id_\CC}: \O_{(\CC,\emptyset)}\to \O_{\ovr{\CC}}$.
By a direct comparison, we observe that this morphism coincides 
with $p_{\ovr{\CC}}$ given in Definition \ref{defi:Rperp}. 
This simple observation will be useful later on.
\end{rem}

\subsection{\label{subsec:gensrels}Presentation by generators and relations}
The aim of this subsection is to show that the colored operads $\O_\CC$ and $\O_{\ovr{\CC}}$
introduced in Section \ref{subsec:SetOperads} have a convenient
presentation in terms of generators and relations, which all admit a natural 
quantum field theoretic interpretation. This is the key to characterize their algebras, 
see Section \ref{subsec:qftalgebras} below. We decided to use a simple and intuitive 
graphical approach to present the generators and relations. The corresponding
colored sequences for the formal construction in Remark \ref{rem:genrel} can be easily 
extracted from these pictures. We start with a generators and relations presentation of 
the auxiliary colored operad $\O_{\CC}$ described in Definition \ref{def:OCC}.
\begin{itemize}
\item {\em Generators $G_\CC$:} We introduce three types of generators
\begin{flalign}\label{eqn:genpics}
\begin{tikzpicture}[cir/.style={circle,draw=black,inner sep=0pt,minimum size=2mm},
        poin/.style={circle, inner sep=0pt,minimum size=0mm}]
%
% Hom
%
\node[poin] (Hout) [label=above:{\small $c^\prime$}] at (0,1) {};
\node[poin] (Hin) [label=below:{\small $c$}] at (0,0) {};
\draw[thick] (Hin) -- (Hout) node[midway,left] {{\small $f$}};
%
% Unit
%
\node[poin] (Uout) [label=above:{\small $t$}] at (4,1) {};
\node[poin] (Uin) [label=below:{\small $\emptyset$}] at (4,0) {};
\draw[thick] (Uin) -- (Uout) node[midway,left] {{\small $1_t$}};
\draw[thick,fill=white] (Uin) circle (0.6mm);
%
% Product
%
\node[poin] (Mout) [label=above:{\small $t$}] at (8,1) {};
\node[poin] (Min1) [label=below:{\small $t$}] at (7.7,0) {};
\node[poin] (Min2) [label=below:{\small $t$}] at (8.3,0) {};
\node[poin] (V)  [label=right:{\small $\mu_t$}] at (8,0.5) {};
%\draw[fill=black] (V) circle (0.6mm);
\draw[thick] (Min1) -- (V);
\draw[thick] (Min2) -- (V);
\draw[thick] (V) -- (Mout);
\end{tikzpicture}
\end{flalign}
for every $t\in\CC_0$ and $(f:c\to c^\prime)\in\mathrm{Mor}\,\CC$.
The quantum field theoretic interpretation is as follows: 
For every spacetime embedding  $f:c\to c^\prime$,
we introduce a $1$-ary operation $f$ to push forward observables.
Moreover, for every spacetime $t$, we introduce a $0$-ary operation $1_t$ and a $2$-ary operation
$\mu_t$ to obtain a unit and a product for the observables on $t$.

\item {\em Functoriality relations $R_{\mathrm{Fun}}$:} We impose the relations
\begin{flalign}\label{eqn:funrel}
\begin{tikzpicture}[cir/.style={circle,draw=black,inner sep=0pt,minimum size=2mm},
        poin/.style={circle, inner sep=0pt,minimum size=0mm}]
%
% Identity Hom
%
%\node[poin] (label1) [label=center:{$\ast_t:$}] at (-1,0.5) {};
\node[poin] (Hout) [label=above:{\small $t$}] at (0,1) {};
\node[poin] (Hin) [label=below:{\small $t$}] at (0,0) {};
\draw[thick] (Hin) -- (Hout) node[midway,left] {{\small $\oone$}};
\node[poin] (EqH) at (0.5,0.5) {$=$};
\node[poin] (HHout) [label=above:{\small $t$}] at (1.5,1) {};
\node[poin] (HHin) [label=below:{\small $t$}] at (1.5,0) {};
\draw[thick] (HHin) -- (HHout) node[midway,left] {{\small $\id_t$}};
%
% Composition Hom
%
%\node[poin] (label2) [label=center:{$(f,g):$}] at (3.5,0.5) {};
\node[poin] (Kout) [label=above:{\small $c^{\prime\prime}$}] at (5,1.25) {};
\node[poin] (Kin) [label=below:{\small $c$}] at (5,-0.25) {};
\node[poin] (Kmid) at (5,0.5) {};
%\draw[fill=black] (Kmid) circle (0.6mm);
\draw[thick] (Kin) -- (Kmid) node[midway,left] {{\small $g$}};
\draw[thick] (Kmid) -- (Kout) node[midway,left] {{\small $f$}};
\node[poin] (EqK) at (5.5,0.5) {$=$};
\node[poin] (KKout) [label=above:{\small $c^{\prime\prime}$}] at (6.5,1) {};
\node[poin] (KKin) [label=below:{\small $c$}] at (6.5,0) {};
\draw[thick] (KKin) -- (KKout) node[midway,left] {{\small $f\,g$}};
\end{tikzpicture}
\end{flalign}
for every $t\in\CC_0$ and every pair of composable morphisms
$(f:c^\prime\to c^{\prime\prime},g:c\to c^\prime)\in 
\mathrm{Mor}\,\CC \, {}_\mathrm{s}^{}\!\times\!{}^{}_\mathrm{t} \, \mathrm{Mor}\,\CC$.
These relations capture the functoriality of pushing forward observables 
by identifying  the operadic unit $\oone$ with the identity morphisms $\id_t$
and operadic compositions of spacetime embeddings with their categorical compositions in $\CC$.

\item {\em Algebra relations $R_{\mathrm{Alg}}$:} We impose the relations 
\begin{flalign}\label{eqn:monrel}
\begin{tikzpicture}[cir/.style={circle,draw=black,inner sep=0pt,minimum size=2mm},
        poin/.style={circle, inner sep=0pt,minimum size=0mm}]
%
% Unitality
%
%\node[poin] (label1) [label=center:{$l_t,\, r_t:$}] at (-1.5,0.5) {};
\node[poin] (Mout) [label=above:{\small $t$}] at (0,1) {};
\node[poin] (Min1) [label=below:{\small $\emptyset$}] at (-0.3,0) {};
\node[poin] (Min2) [label=below:{\small $t$}] at (0.3,0) {};
\node[poin] (V)  [label=right:{\small $\mu_t$}] at (0,0.5) {};
\draw[thick] (Min1) -- (V) node[pos=0.7,left] {{\small $1_t$}};
\draw[thick] (Min2) -- (V);
\draw[thick] (V) -- (Mout);
\draw[thick,fill=white] (Min1) circle (0.6mm);
\node[poin] (EqM) at (1,0.5) {$=$};
\node[poin] (MMout) [label=above:{\small $t$}] at (1.75,1) {};
\node[poin] (MMin) [label=below:{\small $t$}] at (1.75,0) {};
\draw[thick] (MMin) -- (MMout) node[midway,left] {{\small $\oone$}};
\node[poin] (EqMM) at (2.25,0.5) {$=$};
\node[poin] (MMMout) [label=above:{\small $t$}] at (3.25,1) {};
\node[poin] (MMMin1) [label=below:{\small $t$}] at (2.95,0) {};
\node[poin] (MMMin2) [label=below:{\small $\emptyset$}] at (3.55,0) {};
\node[poin] (VVV)  [label=left:{\small $\mu_t$}] at (3.25,0.5) {};
\draw[thick] (MMMin1) -- (VVV);
\draw[thick] (MMMin2) -- (VVV) node[pos=0.7,right] {{\small $1_t$}};
\draw[thick] (VVV) -- (MMMout);
\draw[thick,fill=white] (MMMin2) circle (0.6mm);
%
% Associativity
%
%\node[poin] (label2) [label=center:{$a_t:$}] at (5.5,0.5) {};
\node[poin] (Nout) [label=above:{\small $t$}] at (7,1.25) {};
\node[poin] (Nin1) [label=left:{\small $\mu_t$}] at (6.7,0.25) {};
\node[poin] (Nin2) at (7.3,0) {};
\node[poin] (NV)  [label=right:{\small $\mu_t$}] at (7,0.75) {};
\node[poin] (Ninin1) [label=below:{\small $t$}]  at (6.4,-0.25) {};
\node[poin] (Ninin2) [label=below:{\small $t$}]  at (7,-0.25) {};
\node[poin] (Ninin3) [label=below:{\small $t$}]  at (7.6,-0.25) {};
\draw[thick] (Nin1) -- (NV);
\draw[thick] (Ninin3) -- (NV);
\draw[thick] (Ninin1) -- (Nin1);
\draw[thick] (Ninin2) -- (Nin1);
\draw[thick] (NV) -- (Nout);
\node[poin] (EqN) at (8.25,0.5) {$=$};
\node[poin] (NNout) [label=above:{\small $t$}] at (9.5,1.25) {};
\node[poin] (NNin1) at (9.2,0.25) {};
\node[poin] (NNin2) [label=right:{\small $\mu_t$}] at (9.8,0.25) {};
\node[poin] (NNV)  [label=left:{\small $\mu_t$}] at (9.5,0.75) {};
\node[poin] (NNinin1) [label=below:{\small $t$}]  at (8.9,-0.25) {};
\node[poin] (NNinin2) [label=below:{\small $t$}]  at (9.5,-0.25) {};
\node[poin] (NNinin3) [label=below:{\small $t$}]  at (10.1,-0.25) {};
\draw[thick] (NNinin1) -- (NNV);
\draw[thick] (NNin2) -- (NNV);
\draw[thick] (NNV) -- (NNout);
\draw[thick] (NNinin2) -- (NNin2);
\draw[thick] (NNinin3) -- (NNin2);
\end{tikzpicture}
\end{flalign}
for every $t\in\CC_0$.
These relations capture the unitality and associativity property
of the unit $1_t$ and product $\mu_t$ of observables.

\item {\em Compatibility relations $R_{\mathrm{FA}}$:} We impose the relations
\begin{flalign}\label{eqn:FMrel}
\begin{tikzpicture}[cir/.style={circle,draw=black,inner sep=0pt,minimum size=2mm},
        poin/.style={circle, inner sep=0pt,minimum size=0mm}]
%
% Unit compatibility
%
%\node[poin] (label1) [label=center:{$\langle 1, f \rangle:$}] at (-1.5,0.5) {};
\node[poin] (Kout) [label=above:{\small $c^{\prime}$}] at (0,1.25) {};
\node[poin] (Kin) [label=below:{\small $\emptyset$}] at (0,-0.25) {};
\node[poin] (Kmid) at (0,0.5) {};
\draw[thick] (Kin) -- (Kmid) node[midway,left] {{\small $1_c$}};
\draw[thick] (Kmid) -- (Kout) node[midway,left] {{\small $f$}};
\draw[thick,fill=white] (Kin) circle (0.6mm);
\node[poin] (EqK) at (0.5,0.5) {$=$};
\node[poin] (KKout) [label=above:{\small $c^{\prime}$}] at (1.5,1) {};
\node[poin] (KKin) [label=below:{\small $\emptyset$}] at (1.5,0) {};
\draw[thick] (KKin) -- (KKout) node[midway,left] {{\small $1_{c^\prime}$}};
\draw[thick,fill=white] (KKin) circle (0.6mm);
%
% Product compatibility
%
%\node[poin] (label2) [label=center:{$\langle \mu, f \rangle:$}] at (3.5,0.5) {};
\node[poin] (MMout) [label=above:{\small $c^\prime$}] at (5,1.25) {};
\node[poin] (MMin1) [label=below:{\small $c$}] at (4.5,-0.25) {};
\node[poin] (MMin2) [label=below:{\small $c$}] at (5.5,-0.25) {};
\node[poin] (VV)  [label=right:{\small $\mu_{c}$}] at (5,0.5) {};
\draw[thick] (MMin1) -- (VV);
\draw[thick] (MMin2) -- (VV);
\draw[thick] (VV) -- (MMout)  node[midway,left] {{\small $f$}};
\node[poin] (EqM) at (6.25,0.5) {$=$};
\node[poin] (Mout) [label=above:{\small $c^\prime$}] at (7.5,1.25) {};
\node[poin] (Min1) [label=below:{\small $c$}] at (7,-0.25) {};
\node[poin] (Min2) [label=below:{\small $c$}] at (8,-0.25) {};
\node[poin] (V)  [label=right:{\small $\mu_{c^\prime}$}] at (7.5,0.5) {};
\draw[thick] (Min1) -- (V) node[midway,left] {{\small $f$}};
\draw[thick] (Min2) -- (V) node[midway,right] {{\small $f$}};
\draw[thick] (V) -- (Mout);
\end{tikzpicture}
\end{flalign}
for every $(f:c\to c^\prime)\in\mathrm{Mor}\,\CC$.
These relations capture the compatibility between pushing forward observables
along spacetime embeddings and the unit and product. 
\end{itemize}

From this graphical presentation of the generators and relations 
one easily extracts a pair of parallel $\Seq_{\CC_0}$-morphisms 
\begin{flalign}\label{eqn:relationdecomposition}
r_{\CC,i}\,:=\, r_{\mathrm{Fun},i}\sqcup r_{\mathrm{Alg},i}\sqcup r_{\mathrm{FA},i}~:~
R_\CC \,:=\, R_{\mathrm{Fun}}\sqcup R_{\mathrm{Alg}}\sqcup R_{\mathrm{FA}}
~~\longrightarrow~~F(G_\CC)\quad,
\end{flalign}
for $i=1,2$, which simultaneously captures all three types of relations in the free 
$\CC_0$-colored operad $F(G_\CC)$.
There further exists an $\Op_{\CC_0}$-morphism
$q_{\CC}: F(G_\CC) \to \O_\CC$ to  our colored operad $\O_{\CC}$ from
Definition \ref{def:OCC}, which is defined by setting for the generators
\begin{flalign}\label{eqn:qCCmap}
q_{\CC} \,: \, \left(\begin{array}{ccc}
1_t&\longmapsto&(e,\ast)\\
f&\longmapsto&(e,f)\\
\mu_t&\longmapsto&\big(e,(\id_t,\id_t)\big)
\end{array}\right)\quad.
\end{flalign}
Our first main result of this subsection is
\begin{propo}\label{propo:OCCviagenrel}
Let $\CC$ be a small category. Then 
\begin{flalign}\label{eqn:OCCviagenrel}
\xymatrix@C=4em{
F(R_\CC)~ \ar@<0.7ex>[r]^-{r_{\CC,1}}\ar@<-0.7ex>[r]_-{r_{\CC,2}}&~ F(G_\CC) \ar[r]^-{q_{\CC}}~&~ \O_{\CC}
}
\end{flalign}
is a coequalizer in $\Op_{\CC_0}$.
\end{propo}
\begin{proof}
Verifying that $q_{\CC}\, r_{\CC,1} = q_{\CC}\, r_{\CC,2}$ 
is a simple calculation using 
\eqref{eqn:funrel}, \eqref{eqn:monrel} and \eqref{eqn:FMrel}, as well as 
$q_\CC$ \eqref{eqn:qCCmap}.
As an illustration, let us show this for the second relation in \eqref{eqn:funrel},
\begin{flalign}
q_\CC\big(\gamma(f , g)\big)=\gamma\big(q_\CC(f),q_\CC(g)\big)=
\gamma\big((e,f), (e,g)\big) = (e,f\,g) = q_\CC(f\,g)\quad.
\end{flalign}
It thus remains to prove that 
$q_\CC$ in \eqref{eqn:OCCviagenrel} has the universal property 
of a coequalizer.
\sk

As a crucial preparation, we show that all operations in $\O_\CC$ 
admit a factorization in terms of images (under $q_\CC$) 
of the generators (cf.\ \eqref{eqn:qCCmap}).
Let us define recursively the elements 
$\mu_t^n \in F(G_\CC)\big(\substack{t\\ t^n}\big)$ by
\begin{flalign}
\mu_t^0 \,:=\, 1_t~~, \quad \mu_t^n \,:=\, \gamma\big(\mu_t , ( \oone, \mu_t^{n-1} ) \big)\quad,
\end{flalign}
and observe that 
\begin{flalign}
q_\CC (\mu_t^n) ~=~ 
\big(e,\underbrace{(\id_t,\dots,\id_t)}_{\text{$n$ times}}\big)\,\in\,\O_\CC\big(\substack{t\\ t^n}\big)\quad.
\end{flalign}
For all $t \in \CC_0$, $n \geq 0$, $\und{c}\in \CC_0^n$, 
$(\sigma,\und{f}) \in \O_\CC\big(\substack{t\\ \und{c}}\big)$, 
using the operad structure on $\O_\CC$ (cf.\ Definition \ref{def:OCC}),
we obtain a factorization
\begin{flalign}\label{eqn:elementfact}
(\sigma,\und{f}) ~=~ \O_\CC(\sigma)\, \gamma\Big(q_\CC(\mu_t^n) , 
\Big(q_\CC \big( f_{\sigma^{-1}(1)} \big),\dots , q_\CC \big( f_{\sigma^{-1}(n)}  \big)\Big) \Big)\quad.
\end{flalign}
This implies that
each component $q_\CC : F(G_\CC)\big(\substack{t\\\und{c}}\big) \to \O_\CC\big(\substack{t\\\und{c}}\big)$
is a surjective map of sets.
\sk

We now prove the universal property. Given any $\Op_{\CC_0}$-morphism
$\phi : F(G_\CC)\to \P$ such that $\phi\, r_{\CC,1} = \phi\, r_{\CC,2}$,
we have to construct a unique factorization $\phi = \phi^\prime \, q_\CC$ 
of $\phi$ through $q_\CC : F(G_\CC)\to \O_\CC$.
Uniqueness is immediate because  $q_\CC$ is component-wise surjective.
To prove existence, we use \eqref{eqn:elementfact} and define
\begin{flalign}
\phi^\prime\,:\, \O_\CC\big(\substack{t\\\und{c}}\big)~\longrightarrow~\P\big(\substack{t\\\und{c}}\big)~~,
\quad (\sigma,\und{f})~\longmapsto~\P(\sigma)\, \gamma\Big(\phi(\mu_t^n), \Big(
\phi \big( f_{\sigma^{-1}(1)} \big),\dots, \phi \big( f_{\sigma^{-1}(n)} \big) \Big)\Big)\quad,
\end{flalign}
for all $t \in \CC_0$, $n \geq 0$, $\und{c} \in \CC_0^n$. 
One easily observes that this defines a $\Seq_{\CC_0}$-morphism 
such that $\phi = \phi^\prime \, q_\CC$.
By an elementary but slightly lengthy calculation one shows that 
$\phi^\prime$ is an $\Op_{\CC_0}$-morphism, i.e.\ that it preserves 
permutation actions, compositions and units. 
(Hint: Due to the equivariance axioms of colored operads, 
it is enough for these checks to consider operations $(e,\und{f})$ 
with trivial permutations, which leads to drastic simplifications.)
\end{proof}

Let us now consider the colored operad $\O_{\ovr{\CC}}$ from
Definition \ref{defi:Rperp} and Proposition \ref{propo:OCCbarexplicit},
which is obtained from  $\O_\CC$ by enforcing relations to 
capture the $\perp$-commutativity axiom. 
For our presentation by generators and relations, 
this amounts to adding one more type of relations.
\begin{itemize}
\item {\em $\perp$-commutativity relations $R_\perp$:} We impose the relations
\begin{flalign}\label{eqn:perprel}
\begin{tikzpicture}[cir/.style={circle,draw=black,inner sep=0pt,minimum size=2mm},
        poin/.style={circle, inner sep=0pt,minimum size=0mm}]
%
%
%
%\node[poin] (label) [label=center:{$(f_1,f_2):$}] at (-2,0.5) {};
\node[poin] (Mout) [label=above:{\small $c$}] at (0,1.25) {};
\node[poin] (Min1) [label=below:{\small $c_1$}] at (-0.5,-0.25) {};
\node[poin] (Min2) [label=below:{\small $c_2$}] at (0.5,-0.25) {};
\node[poin] (V)  [label=right:{\small $\mu_c$}] at (0,0.5) {};
\draw[thick] (Min1) -- (V) node[midway,left] {{\small $f_1$}};
\draw[thick] (Min2) -- (V) node[midway,right] {{\small $f_2$}};
\draw[thick] (V) -- (Mout);
\node[poin] (Eq) at (1.25,0.5) {$=$};
\node[poin] (MMout) [label=above:{\small $c$}] at (2.5,1.5) {};
\node[poin] (MMin1)  at (2,0) {};
\node[poin] (MMin2)  at (3,0) {};
\node[poin] (MMin11) [label=below:{\small $c_1$}] at (2,-0.5) {};
\node[poin] (MMin21) [label=below:{\small $c_2$}] at (3,-0.5) {};
\node[poin] (VV)  [label=right:{\small $\mu_c$}] at (2.5,0.75) {};
\draw[thick] (MMin1) -- (VV) node[midway,left] {{\small $f_2$}};
\draw[thick] (MMin2) -- (VV) node[midway,right] {{\small $f_1$}};
\draw[thick] (VV) -- (MMout);
\draw[thick] (MMin11) -- (MMin2);
\draw[thick] (MMin21) -- (MMin1);
\end{tikzpicture}
\end{flalign}
for every orthogonal pair of $\CC$-morphisms $(f_1:c_1\to c,f_2:c_2\to c)\in {\perp}$.
These relations capture the commutative behavior of
pairs of observables on $c$ which arise from pushforwards 
along an orthogonal pair of spacetime embeddings. 
\end{itemize}

To capture simultaneously all four types of relations (functoriality, algebra, compatibility 
and $\perp$-commutativity relations), we consider the pair of parallel $\Seq_{\CC_0}$-morphism 
\begin{flalign}\label{eqn:decompositionRovrCC}
r_{\ovr{\CC},i}\, :=\,  r_{\CC,i}\sqcup r_{\perp,i} ~:~ 
R_{\ovr{\CC}} ~:=~ R_\CC \sqcup R_\perp
~~\longrightarrow~~F(G_\CC)\quad,
\end{flalign}
for $i=1,2$. Composing $q_\CC$ given in \eqref{eqn:qCCmap} with the 
$\Op_{\CC_0}$-morphism $p_{\ovr{\CC}} : \O_\CC\to \O_{\ovr{\CC}}$ 
given in Definition \ref{defi:Rperp}, we obtain an $\Op_{\CC_0}$-morphism
\begin{flalign}
p_{\ovr{\CC}}\,q_\CC \,:\, F(G_\CC)~\longrightarrow ~ \O_{\ovr{\CC}}\quad.
\end{flalign}
Our second main result of this subsection is
\begin{theo}\label{theo:OCCperpviagenrel}
Let $\ovr{\CC}=(\CC,\perp)$ be an orthogonal category.  Then 
\begin{flalign}\label{eqn:OCCperpviagenrel}
\xymatrix@C=4em{
F(R_{\ovr{\CC}})~ \ar@<0.7ex>[r]^-{r_{\ovr{\CC},1}}\ar@<-0.7ex>[r]_-{r_{\ovr{\CC},2}}&~ F(G_\CC) \ar[r]^-{p_{\ovr{\CC}}\, q_{\CC}}~&~ \O_{\ovr{\CC}}
}
\end{flalign}
is a coequalizer in $\Op_{\CC_0}$.
\end{theo}
\begin{proof}
This is a simple consequence of Definition \ref{defi:Rperp}, 
Proposition \ref{propo:OCCviagenrel} and the diagram
\begin{flalign}\label{tmp:diagrammanycoeq}
\xymatrix@C=4em@R=4em{
\ar[d]_-{\iota_1}F(R_\CC)  \ar@<0.9ex>[rd]^-{r_{\CC,1}}\ar@<-0.5ex>[rd]_-{r_{\CC,2}}  ~&~ ~&~ ~&~ \P \\
F(R_{\CC})\sqcup F(R_{\perp}) \ar@<0.7ex>[r]^-{r_{\ovr{\CC},1}}\ar@<-0.7ex>[r]_-{r_{\ovr{\CC},2}} ~&~ F(G_\CC) \ar@<0.5ex>[rru]^-{\phi} \ar[r]^-{q_{\CC}} ~&~ \O_\CC\ar[r]^-{p_{\ovr{\CC}}} \ar@{-->}[ru]_-{\exists!\,\phi^\prime}~&~ \O_{\ovr{\CC}} \ar@{-->}[u]_-{\exists!\, \phi^{\prime\prime}}\\
\ar[u]^-{\iota_2}F(R_\perp)  \ar@<-0.2ex>[rru]^-{s_{\perp,1}}\ar@<-1.6ex>[rru]_-{s_{\perp,2}} ~&~ ~&~ ~&~ 
}
\end{flalign}
in $\Op_{\CC_0}$, where $\phi : F(G_\CC)\to \P$ is 
any $\Op_{\CC_0}$-morphism such that 
$\phi\, r_{\ovr{\CC},1} = \phi\, r_{\ovr{\CC},2}$.
Let us explain this in more detail:
(1)~We have that $F(R_{\ovr{\CC}}) \cong F(R_{\CC})\sqcup F(R_{\perp})$
because $F$ is a left adjoint functor, hence it preserves coproducts.
(2)~By definition of $r_{\ovr{\CC},i}$ the upper left triangles commute, 
i.e.\ $r_{\CC,i} = r_{\ovr{\CC},i}\, \iota_1$ for $i=1,2$.
(3)~By a short calculation using \eqref{eqn:perprel}, \eqref{eqn:qCCmap}
and \eqref{eqn:rperp} one shows that the lower left triangles commute, 
i.e.\ $s_{\perp,i} = q_\CC \,r_{\ovr{\CC},i}\, \iota_2$ for $i=1,2$.
(4)~Because of (2), $\phi\, r_{\CC,1} = \phi\, r_{\CC,2}$.
Existence and uniqueness of the morphism $\phi^\prime : \O_\CC\to \P$ in
\eqref{tmp:diagrammanycoeq} is then a consequence of Proposition \ref{propo:OCCviagenrel}.
(5)~Because of (3), $\phi^\prime : \O_\CC\to \P$ coequalizes $s_{\perp,1}$ and $s_{\perp,2}$. 
Existence and uniqueness of the morphism $\phi^{\prime\prime} : \O_{\ovr{\CC}}  \to \P$ in
\eqref{tmp:diagrammanycoeq} is then a consequence of Definition \ref{defi:Rperp}.
\end{proof}

We conclude this subsection by noting that our constructions are easily seen to be functorial.
\begin{propo}\label{propo:functorialityofgenrel}
The coequalizer \eqref{eqn:OCCviagenrel} is natural in $\CC\in\Cat$
and the coequalizer \eqref{eqn:OCCperpviagenrel} is natural in $\ovr{\CC}=(\CC,\perp) \in\OCat$.
\end{propo}

\subsection{\label{subsec:qftalgebras}Algebras}
The following result is an immediate consequence of
our presentation by generators and relations in Theorem \ref{theo:OCCperpviagenrel}
of the colored operad $\O_{\ovr{\CC}}$.
\begin{theo}\label{theo:AlgoverOCCperp}
Let $\ovr{\CC}=(\CC,\perp)\in\OCat$ be an orthogonal category.
Then the category $\Alg_{\O_{\ovr{\CC}}}(\MM)$ of $\O_{\ovr{\CC}}$-algebras in $\MM$ 
is the category $\Alg(\MM)^{\ovr{\CC}}$ of 
$\perp$-commutative functors from $\CC$ 
to  associative and unital algebras in $\MM$, cf.\ Definition \ref{def:perpcommutative}.
\end{theo}

\begin{rem}\label{rem:AlgoverOCC}
Let $\CC$ be a small category and consider the 
orthogonal category $(\CC,\emptyset)$
with trivial orthogonality relation. By Remark \ref{rem:pCCviafunctoriality},
we have that $\O_{(\CC,\emptyset)} = \O_\CC\in\Op$ is our auxiliary operad
from Definition \ref{def:OCC}. As a special instance
of Theorem  \ref{theo:AlgoverOCCperp}, we obtain 
that $\Alg_{\O_\CC}(\MM)$ is the 
category $\Alg(\MM)^\CC$ of all 
functors from $\CC$ to associative and unital 
algebras in $\MM$.
\end{rem}

\begin{rem}\label{rem:algebrasphysics}
There is the following quantum field theoretic interpretation.
As in Section \ref{subsec:motivationQFT}, we interpret $\ovr{\CC}= (\CC,\perp)$
as a category of spacetimes $\CC$, together with a specification $\perp$ of pairs 
of subspacetimes for which observables are supposed to behave commutatively. The category
$\Alg(\MM)^{\ovr{\CC}}$ of $\perp$-commutative functors describes all 
possible quantum field theories for this scenario. Theorem \ref{theo:AlgoverOCCperp}
deepens our understanding of the algebraic structures underlying such 
quantum field theories by proving that they are precisely the algebras over
our colored operad $\O_{\ovr{\CC}}\in \Op$. 
It is worth to emphasize the following analogy: 
One of the key ideas of algebraic quantum field theory is to shift
the focus from (Hilbert space) representations of algebras to the underlying 
abstract algebras in order to analyze structural properties of quantum field theories.
Our operadic framework goes one level deeper by shifting the focus
from specific realizations of the algebraic structures of quantum field theories 
to the underlying abstract operads.
\end{rem}

\subsection{\label{subsec:examples}Examples}
We present explicit examples of orthogonal categories
$\ovr{\CC}=(\CC,\perp)\in\OCat$ that are motivated by algebraic quantum field theory. 
For this purpose we shall need the following constructions. The first one is straightforward to verify.
\begin{lem}\label{lem:pullpushorth}
Let $F :\CC\to\DD$ be a functor between small categories.
\begin{itemize}
\item[(i)] Given any orthogonality relation 
${\perp}_{\DD}\subseteq \mathrm{Mor}\,\DD \, {}_\mathrm{t}^{}\!\times\!{}^{}_\mathrm{t} \,\mathrm{Mor}\,\DD$
on $\DD$, then
\begin{flalign}
F^\ast(\perp_\DD) \,:=\, \Big\{ (f_1,f_2) \in 
 \mathrm{Mor}\,\CC \, {}_\mathrm{t}^{}\!\times\!{}^{}_\mathrm{t} \,\mathrm{Mor}\,\CC ~:~F(f_1)\perp_\DD F(f_2)\Big\}
\end{flalign}
is an orthogonality relation on $\CC$. We call $F^\ast(\perp_\DD)$ the pullback
of $\perp_\DD$ along $F$ and note that $F :  (\CC,F^\ast(\perp_\DD))
\to (\DD,\perp_\DD)$ is an orthogonal functor.

\item[(ii)] Given any orthogonality relation ${\perp}_\CC \subseteq
\mathrm{Mor}\,\CC \, {}_\mathrm{t}^{}\!\times\!{}^{}_\mathrm{t} \,\mathrm{Mor}\,\CC$
on $\CC$, then
\begin{multline}
F_{\ast}(\perp_\CC) \,:=\, \Big\{ \big(g\,F(f_1)\,h_1,g\,F(f_2)\,h_2\big) \in  
\mathrm{Mor}\,\DD \, {}_\mathrm{t}^{}\!\times\!{}^{}_\mathrm{t} \,\mathrm{Mor}\,\DD ~:~\\
f_1\perp_\CC f_2 \text{ and }g,h_1,h_2\in\mathrm{Mor}\,\DD \text{ are composable}\Big\}
\end{multline}
is an orthogonality relation on $\DD$. We call $F_\ast(\perp_\CC)$ the pushforward
of $\perp_\CC$ along $F$ and note that $F : (\CC,\perp_\CC)
\to (\DD,F_\ast(\perp_\CC))$ is an orthogonal functor.
\end{itemize}
\end{lem}

\begin{lem}\label{lem:localizedperpcom}
Let $\ovr{\CC}=(\CC,\perp)\in\OCat$ and $W\subseteq \mathrm{Mor}\,\CC$
any subset. Define $\ovr{\CC[W^{-1}]}:= (\CC[W^{-1}],L_\ast(\perp))\in\OCat$
by pushing forward the orthogonality relation $\perp$ on $\CC$
along the localization functor $L : \CC\to\CC[W^{-1}]$.
There exists a canonical identification between
\begin{itemize}
\item[(1)] $\perp$-commutative and $W$-constant functors $\mathfrak{B} : \CC\to \Alg(\MM)$,
\item[(2)] $L_\ast(\perp)$-commutative functors $\mathfrak{A} : \CC[W^{-1}] \to \Alg(\MM)$.
\end{itemize} 
\end{lem}
\begin{proof}
Given any $L_\ast(\perp)$-commutative functor $\mathfrak{A} : \CC[W^{-1}] \to \Alg(\MM)$,
define $\mathfrak{B} := \mathfrak{A}\,L : \CC\to\Alg(\MM)$ which is 
obviously $W$-constant and $\perp$-commutative.
\sk

Conversely, let  $\mathfrak{B} : \CC\to\Alg(\MM)$ be any $\perp$-commutative and $W$-constant functor.
Define $\mathfrak{A} : \CC[W^{-1}]\to\Alg(\MM)$ via
the universal property of localizations, i.e.\ $\mathfrak{A}$ is the unique functor
such that $\mathfrak{A}\,L =\mathfrak{B}$. We have to show that
$\mathfrak{A}$ is $L_\ast(\perp)$-commutative.
By definition of the pushforward orthogonality relation,
any element in $L_\ast(\perp)$ may be written as 
$\big(g\,L(f_1)\,h_1, g\, L(f_2)\,h_2\big)$ with $f_1\perp f_2$ and 
$g,h_1,h_2\in\mathrm{Mor}\,\CC[W^{-1}]$. Using that $\mathfrak{A} : \CC[W^{-1}]\to\Alg(\MM)$ 
is a functor, the diagram in Definition \ref{def:perpcommutative} corresponding to
$\big(g\,L(f_1)\,h_1, g\, L(f_2)\,h_2\big)$ may be decomposed into five smaller squares.
One observes that it suffices to prove that the square
\begin{flalign}
\xymatrix@C=3em{
 \ar[d]_-{\mathfrak{A}\,L (f_1)\otimes \mathfrak{A}\,L (f_2) }\mathfrak{A}(c_1)\otimes \mathfrak{A}(c_2) \ar[rr]^-{\mathfrak{A}\,L (f_1)\otimes \mathfrak{A}\,L (f_2) } ~&&~\mathfrak{A}(c)\otimes\mathfrak{A}(c)\ar[d]^-{\mu_c^\op}\\
\mathfrak{A}(c)\otimes\mathfrak{A}(c)\ar[rr]_-{\mu_c}~&&~\mathfrak{A}(c)
}
\end{flalign}
commutes, which is true because $\mathfrak{A}\,L =\mathfrak{B}$ is by hypothesis
$\perp$-commutative and $f_1\perp f_2$.
\end{proof}

\begin{ex}[Locally covariant quantum field theory without time-slice axiom]\label{ex:LCQFTnoTS}
In locally covariant quantum field theory \cite{Brunetti,FewsterVerch}
one considers the category $\Loc$ of oriented, time-oriented
and globally hyperbolic Lorentzian manifolds of a fixed dimension $m\geq 2$.
Concretely, the objects in $\Loc$ are tuples $\mathbb{M}=(M,g,\mathfrak{o},\mathfrak{t})$
where $M$ is an $m$-dimensional manifold (Hausdorff and second-countable), 
$g$ is a globally hyperbolic Lorentzian metric on $M$, $\mathfrak{o}$ 
is an orientation of $M$ and $\mathfrak{t}$ is a time-orientation
of $(M,g)$. A $\Loc$-morphism $f: \mathbb{M}\to 
\mathbb{M}^\prime$ is an orientation
and time-orientation preserving isometric embedding,
such that the image $f(M)\subseteq M^\prime$ is causally convex and open.
(For an introduction to Lorentzian geometry we refer the reader to e.g.\ \cite{BGP}.)
Note that the category $\Loc$ is {\em not} small, however it is equivalent
to a small category, i.e.\ it is essentially small. As usual, this follows by using
Whitney's Embedding Theorem to realize (up to diffeomorphism) all $m$-dimensional
manifolds $M$ as submanifolds of $\bbR^{2m+1}$. In the following we always choose a
small category equivalent to $\Loc$ and denote it with abuse of notation
also by $\Loc$. Our results in Section \ref{subsec:equivalences} imply
that it does not matter which small subcategory equivalent to $\Loc$ we choose.
More precisely, different choices define equivalent categories of algebras over their 
associated colored operads.
\sk

We equip the category $\Loc$ with the following orthogonality relation:
Two $\Loc$-morphisms $f_1 : \mathbb{M}_1\to\mathbb{M}$ and $f_2 : \mathbb{M}_2 \to \mathbb{M}$
are orthogonal, $f_1\perp f_2$, if and only if their images $f_1(M_1)$ and $f_2(M_2)$
are causally disjoint subsets in $\mathbb{M}$, i.e.\ $f_1(M_1) \cap J_{\mathbb{M}}(f_2(M_2)) =\emptyset$,
where $J_{\mathbb{M}}(S):= J_{\mathbb{M}}^+(S)\cup J_{\mathbb{M}}^-(S)\subseteq M$ 
denotes the union of the causal future and past of a subset $S\subseteq M$.
It is easy to verify the symmetry and composition stability properties of Definition \ref{def:orthcat}.
Hence, $\ovr{\Loc} := (\Loc,\perp)\in \OCat$ is an orthogonal category.
\sk

By Theorem \ref{theo:AlgoverOCCperp}, we obtain that 
algebras over the colored operad $\O_{\ovr{\Loc}} \in \Op$
are canonically identified with functors $\mathfrak{A} : \Loc \to \Alg(\MM)$ from $\Loc$
to the category of associative and unital algebras in $\MM$ that satisfy the $\perp$-commutativity 
axiom (cf.\ Definition \ref{def:perpcommutative}). These are 
locally covariant quantum field theories \cite{Brunetti,FewsterVerch}
satisfying the Einstein causality axiom, but not necessarily the time-slice axiom.
\sk

An interesting problem, originally considered in 
\cite{Fre1,Fre2,Fre3}, is that of extending 
quantum field theories from a subcategory 
of ``nice'' spacetimes to all spacetimes. 
To formalize this in our framework, let us consider the 
full subcategory $\Loc_\diamond \subseteq \Loc$ of {\em diamond} spacetimes, 
i.e.\ objects $\mathbb{M} \in \Loc$ 
with underlying manifold diffeomorphic to $\bbR^m$, 
and denote the inclusion functor by $j: \Loc_\diamond \to \Loc$. 
We equip $\Loc_\diamond$ with the pullback orthogonality relation 
$j^\ast(\perp)$ and obtain an orthogonal functor 
$j : \overline{\Loc}_\diamond := (\Loc_\diamond,j^\ast(\perp)) 
\to \overline{\Loc}$. Notice that $\O_{\overline{\Loc}_\diamond}$-algebras 
are $j^\ast(\perp)$-commutative functors $\mathfrak{A}: \Loc_\diamond \to \Alg(\MM)$, i.e.\  
locally covariant quantum field theories defined only on diamond spacetimes, cf.\ \cite{Lang}. 
By Proposition \ref{propo:functoriality} and Theorem \ref{theo:coloradjunctionalgebras}, 
the orthogonal functor $j: \ovr{\Loc}_\diamond \to \ovr{\Loc}$ induces an adjunction
\begin{flalign}
\xymatrix{
{\O_{j}}_! \,:\, \Alg_{\O_{\ovr{\Loc}_\diamond}}(\MM) ~\ar@<0.5ex>[r] & \ar@<0.5ex>[l] ~\Alg_{\O_{\ovr{\Loc}}}(\MM) \,:\, \O_{j}^\ast
}\quad.
\end{flalign}
The right adjoint $\O_{j}^\ast$ restricts 
quantum field theories defined on all of $\Loc$ to 
$\Loc_\diamond$. More interestingly, the left adjoint
${\O_{j}}_! $ is an extension functor that extends quantum field theories
defined on $\Loc_\diamond$ to all of $\Loc$.
We shall analyze this adjunction in more detail in Section \ref{subsec:subcats}.
In Section \ref{sec:QFTconstructions}, we also compare our constructions
to Fredenhagen's universal algebra \cite{Fre1,Fre2,Fre3}, which is obtained by 
left Kan extension of algebra-valued functors \cite{Lang}.
\end{ex}

\begin{ex}[Locally covariant quantum field theory with time-slice axiom]\label{ex:LCQFTwithTS}
In the setting of Example \ref{ex:LCQFTnoTS}, recall that a morphism
$f : \mathbb{M} \to \mathbb{M}^\prime$ is called a {\em Cauchy morphism}
if its image $f(M)\subseteq M^\prime$ contains a Cauchy surface of $\mathbb{M}^\prime$.
We denote the subset of all Cauchy morphisms by $W\subseteq \mathrm{Mor}\,\Loc$
and consider the localization $L : \Loc \to \Loc[W^{-1}]$. Using Lemma \ref{lem:pullpushorth},
we may equip $\Loc[W^{-1}]$ with the pushforward orthogonality relation
$L_\ast(\perp)$ and obtain an orthogonal functor $L :\ovr{\Loc}  
\to \big(\Loc[W^{-1}],L_\ast(\perp) \big) =: \ovr{\Loc[W^{-1}]}$.
\sk

By Theorem \ref{theo:AlgoverOCCperp}, we obtain that 
algebras over $\O_{\ovr{\Loc[W^{-1}]}}\in \Op$
are canonically identified with $L_\ast(\perp)$-commutative
functors $\mathfrak{A} : \Loc[W^{-1}] \to \Alg(\MM)$. By Lemma \ref{lem:localizedperpcom}, 
such functors are canonically identified with 
$\perp$-commutative and $W$-constant functors 
$\mathfrak{A} : \Loc \to \Alg(\MM)$ on $\Loc$.
These are locally covariant quantum field theories 
\cite{Brunetti,FewsterVerch} satisfying the Einstein causality axiom and the time-slice axiom.
\sk

Using Proposition \ref{propo:functoriality} and Theorem \ref{theo:coloradjunctionalgebras},
the orthogonal functor $L :\ovr{\Loc} \to \ovr{\Loc[W^{-1}]}$ defines an adjunction
\begin{flalign}
\xymatrix{
{\O_{L}}_! \,:\, \Alg_{\O_{\ovr{\Loc}}}(\MM) ~\ar@<0.5ex>[r] & \ar@<0.5ex>[l]~ \Alg_{\O_{\ovr{\Loc[W^{-1}]}}}(\MM) \,:\, \O_{L}^\ast
}\quad.
\end{flalign}
We call the left adjoint ${\O_{L}}_!$ the {\em $W$-constantification functor}
or, more specifically, the {\em time-slicification functor}. To a quantum
field theory which may not satisfy the time-slice axiom
it assigns one which does. We shall analyze such adjunctions in more detail 
in Section \ref{subsec:localization}.
\end{ex}

\begin{rem}
Mimicking the previous examples, one can  introduce 
further interesting orthogonal categories which give rise to
algebraic quantum field theories on a fixed spacetime \cite{HaagKastler},
chiral conformal quantum field theories \cite{Kawahigashi,Rehren,BDHcft} 
and Euclidean quantum field theories \cite{Schlingemann}. For the latter
two scenarios, the orthogonality relation is determined by disjointness instead
of causal disjointness. Our framework also applies to
spacetimes with boundaries, in which case adjunctions
similar to those above provide interesting insights into boundary conditions
for quantum field theories, see \cite{BDSboundary}.
\end{rem}

%%%%%%%%%%%%%%%%%%%%%%%%%%%%%%%%%%%%%%%%%%%%%%%%
%%%%%%%%%%%%%%%%%%%%%%%%%%%%%%%%%%%%%%%%%%%%%%%%

\section{\label{sec:QFTadjunctions}Algebra adjunctions}
Given an orthogonal functor $F : \ovr{\CC} \to \ovr{\DD}$,
we obtain by Proposition \ref{propo:functoriality} 
an $\Op$-morphism $\O_F: \O_{\ovr{\CC}}\to\O_{\ovr{\DD}}$
and thus by Theorem \ref{theo:coloradjunctionalgebras} an adjunction 
\begin{flalign}\label{eqn:generalorthadjunction}
\xymatrix{
{\O_{F}}_! \,:\, \Alg_{\O_{\ovr{\CC}}}(\MM) ~\ar@<0.5ex>[r] & \ar@<0.5ex>[l]~ \Alg_{\O_{\ovr{\DD}}}(\MM) \,:\, \O_{F}^\ast
}
\end{flalign}
between the categories of algebras. The examples in Section \ref{subsec:examples}
show that such adjunctions lead to interesting constructions in quantum
field theory, for example $W$-constantification/time-slicification (cf.\ Example \ref{ex:LCQFTwithTS})
and local-to-global extensions (cf.\ Example \ref{ex:LCQFTnoTS}).
The aim of this section is to study these adjunctions 
for particularly interesting classes of orthogonal functors in more detail.
We will also explain the significance of our results
for quantum field theory.

\subsection{\label{subsec:general}General orthogonal functors}
For a general orthogonal functor $F$, we establish a relation between
\eqref{eqn:generalorthadjunction} and the adjunction
\begin{flalign}\label{eqn:Lanadjunction}
\xymatrix{
\Lan_F \,:\, \Alg(\MM)^\CC~\ar@<0.5ex>[r]&\ar@<0.5ex>[l]  ~\Alg(\MM)^\DD \,:\, F^\ast
}
\end{flalign}
obtained by left Kan extension of algebra-valued functors along the functor $F:\CC\to\DD$.
Notice that the latter neglects the orthogonality relations on $\CC$ and $\DD$.
\sk

In order to compare these two adjunctions,
let us recall from Remark \ref{rem:AlgoverOCC} 
that $\Alg(\MM)^\CC \cong \Alg_{\O_\CC}(\MM)$, where $\O_\CC\in\Op_{\CC_0}$ is
our auxiliary colored operad that does not encode the $\perp$-commutativity relations. 
By Proposition \ref{propo:functoriality}, there exists 
a natural $\Op_{\CC_0}$-morphism $p_{\ovr{\CC}} : \O_{\CC}\to\O_{\ovr{\CC}}$,
hence we obtain from Theorem \ref{theo:coloradjunctionalgebras} a natural adjunction
\begin{flalign}\label{eqn:pCCAlgadjunction}
\xymatrix{
{p_{\ovr{\CC}}}_! \,:\, \Alg(\MM)^\CC~\ar@<0.5ex>[r]&\ar@<0.5ex>[l]  ~ \Alg_{\O_{\ovr{\CC}}}(\MM) \,:\, p_{\ovr{\CC}}^\ast
}\quad.
\end{flalign}
Theorem  \ref{theo:AlgoverOCCperp} implies that 
$\Alg_{\O_{\ovr{\CC}}}(\MM) \cong \Alg(\MM)^{\ovr{\CC}}$
and it is easy to verify that under this identification the right adjoint
functor $p_{\ovr{\CC}}^\ast$ in \eqref{eqn:pCCAlgadjunction}
is given by the functor $U : \Alg(\MM)^{\ovr{\CC}} \to \Alg(\MM)^\CC$
that forgets $\perp$-commutativity. As the latter is a full subcategory embedding, 
we observe that $ \Alg_{\O_{\ovr{\CC}}}(\MM)$ is a {\em full reflective subcategory}
of $ \Alg(\MM)^{\CC}$, see e.g.\ \cite[Chapter IV.3]{MacLane}.
Summing up, we obtain 
\begin{lem}\label{lem:pbarCreflective}
The natural adjunction \eqref{eqn:pCCAlgadjunction} exhibits
$\Alg_{\O_{\ovr{\CC}}}(\MM)$ as a full reflective subcategory of 
$\Alg(\MM)^{\CC}$, i.e.\ the counit  $\epsilon : {p_{\ovr{\CC}}}_!~ p_{\ovr{\CC}}^\ast 
\to \id_{\Alg_{\O_{\ovr{\CC}}}(\MM)}$ is a natural isomorphism.
\end{lem}

\begin{rem}
The adjunction \eqref{eqn:pCCAlgadjunction} admits a quantum field theoretic
interpretation. The category $\Alg(\MM)^{\CC}$ contains also functors that do not
satisfy the $\perp$-commutativity axiom and hence should not be regarded 
as quantum field theories. The left adjoint functor ${p_{\ovr{\CC}}}_!$ in \eqref{eqn:pCCAlgadjunction} 
allows us to assign to any functor $\mathfrak{B} :\CC\to\Alg(\MM)$ 
a bona fide quantum field theory 
${p_{\ovr{\CC}}}_! (\mathfrak{B}) \in \Alg_{\O_{\ovr{\CC}}}(\MM)$.
This construction may be called {\em $\perp$-abelianization} due to its structural similarity
with the abelianization of algebraic structures such as groups or algebras.
By Lemma \ref{lem:pbarCreflective},
we know that the counit of the adjunction  \eqref{eqn:pCCAlgadjunction} is a natural
isomorphism. Concretely, this means that the $\perp$-abelianization of the functor $\mathfrak{B} = 
p_{\ovr{\CC}}^\ast (A)$ underlying a bona fide quantum field theory
$A \in \Alg_{\O_{\ovr{\CC}}}(\MM)$ is isomorphic to 
$A$ itself via $\epsilon: {p_{\ovr{\CC}}}_!\, p_{\ovr{\CC}}^\ast (A) \to A$. 
\end{rem}

Recalling Remarks \ref{rem:pCCviafunctoriality} and  \ref{rem:coloradjunctioncomposition}, 
we observe that there exists a diagram of adjunctions
\begin{flalign}\label{eqn:LanopLandiagram}
\xymatrix{
\ar@<0.5ex>[d]^-{p_{\ovr{\CC}}^\ast} \Alg_{\O_{\ovr{\CC}}}(\MM) \ar@<0.5ex>[rr]^-{{\O_F}_!} ~&&~ \ar@<0.5ex>[ll]^-{\O_F^\ast} \Alg_{\O_{\ovr{\DD}}}(\MM) \ar@<0.5ex>[d]^-{p_{\ovr{\DD}}^\ast} \\
\ar@<0.5ex>[u]^-{{p_{\ovr{\CC}}}_!} \Alg(\MM)^\CC \ar@<0.5ex>[rr]^-{\Lan_F} ~&&~ \ar@<0.5ex>[ll]^-{F^\ast} \Alg(\MM)^\DD \ar@<0.5ex>[u]^-{{p_{\ovr{\DD}}}_!}
}
\end{flalign} 
in which the square formed by the right adjoint functors commutes,
i.e.\ $p_{\ovr{\CC}}^\ast~\O_F^\ast =F^\ast~ p_{\ovr{\DD}}^\ast$.
This allows us to prove the main result of this subsection.
\begin{propo}\label{propo:causalcorrect}
Let $F:\ovr{\CC} \to \ovr{\DD}$ be an orthogonal functor. There exists
a natural isomorphism
\begin{flalign}
{\O_F}_! ~\cong~ {p_{\ovr{\DD}}}_! ~\Lan_F~ p_{\ovr{\CC}}^\ast
\end{flalign}
of functors $\Alg_{\O_{\ovr{\CC}}}(\MM) \to \Alg_{\O_{\ovr{\DD}}}(\MM)$. 
\end{propo}
\begin{proof}
Let $A \in \Alg_{\O_{\ovr{\CC}}}(\MM)$ and $B \in \Alg_{\O_{\ovr{\DD}}}(\MM)$.
Using \eqref{eqn:LanopLandiagram}, we obtain 
the following chain of natural bijections of Hom-sets
\begin{flalign}
\nn \Alg_{\O_{\ovr{\DD}}}(\MM)\big( {p_{\ovr{\DD}}}_!\, \Lan_F\, p_{\ovr{\CC}}^\ast(A),B\big) &~\cong~
\Alg(\MM)^\DD\big( \Lan_F\, p_{\ovr{\CC}}^\ast(A), p_{\ovr{\DD}}^\ast( B)\big)\\
\nn &~\cong~ \Alg(\MM)^\CC\big(  p_{\ovr{\CC}}^\ast(A), F^\ast\, p_{\ovr{\DD}}^\ast( B)\big)\\
 &~=~ \Alg(\MM)^\CC\big(  p_{\ovr{\CC}}^\ast(A), p_{\ovr{\CC}}^\ast\, \O_F^\ast( B)\big)\quad,
\end{flalign}
where in the last step we used that the square formed by the right adjoint functors commutes.
Using also that the functor  $p_{\ovr{\CC}}^\ast$ is fully faithful (cf.\ Lemma \ref{lem:pbarCreflective}),
we obtain 
\begin{flalign}
\Alg_{\O_{\ovr{\DD}}}(\MM)\big( {p_{\ovr{\DD}}}_!\, \Lan_F p_{\ovr{\CC}}^\ast(A),B\big) ~\cong~
\Alg_{\O_{\ovr{\CC}}}(\MM)\big( A, \O_F^\ast( B)\big)\quad,
\end{flalign}
which implies that ${p_{\ovr{\DD}}}_! \,\Lan_F\, p_{\ovr{\CC}}^\ast$ is a left adjoint of
$\O_F^\ast$. The uniqueness (up to natural isomorphism) of adjoint functors implies the assertion.
\end{proof}

\subsection{\label{subsec:localization}Orthogonal localizations}
Let $\ovr{\CC} = (\CC,\perp)\in\OCat$ be an orthogonal category
and $W\subseteq \mathrm{Mor}\,\CC$ a subset of the set of morphisms.
Consider the localized category $\CC[W^{-1}]$ together with the localization
functor $L : \CC\to\CC[W^{-1}]$. We define an orthogonality relation $L_\ast(\perp)$
on $\CC[W^{-1}]$ by using the pushforward construction from Lemma \ref{lem:pullpushorth}.
We obtain an orthogonal functor $L : \ovr{\CC} \to (\CC[W^{-1}],L_\ast(\perp))=: \ovr{\CC[W^{-1}]}$
and hence by Proposition \ref{propo:functoriality} and Theorem \ref{theo:coloradjunctionalgebras}
an adjunction
\begin{flalign}\label{eqn:localizationadjunction}
\xymatrix{
{\O_{L}}_! \,:\, \Alg_{\O_{\ovr{\CC}}}(\MM) ~\ar@<0.5ex>[r] & \ar@<0.5ex>[l]~ \Alg_{\O_{\ovr{\CC[W^{-1}]}}}(\MM) \,:\, \O_{L}^\ast\quad.
}
\end{flalign}
Recall from Example \ref{ex:LCQFTwithTS} that the left adjoint
is interpreted in terms of 
$W$-constantification/time-slicification. Adjunctions
obtained from orthogonal localizations enjoy the following property.
\begin{propo}\label{propo:timeslicificationproperties}
The adjunction \eqref{eqn:localizationadjunction} exhibits 
$\Alg_{\O_{\ovr{\CC[W^{-1}]}}}(\MM)$
as a full reflective subcategory of $\Alg_{\O_{\ovr{\CC}}}(\MM)$, i.e.\
the counit $\epsilon : {\O_{L}}_!\, \O_{L}^\ast \to \id_{\Alg_{\O_{\ovr{\CC[W^{-1}]}}}(\MM)}$
is a natural isomorphism.
\end{propo}
\begin{proof}
The right adjoint functor $\O_{L}^\ast$ in \eqref{eqn:localizationadjunction}
is given by restricting the pullback functor
\begin{flalign}\label{eqn:Lasttmp}
L^\ast \, :\,  \Alg(\MM)^{\CC[W^{-1}]}~\longrightarrow~\Alg(\MM)^{\CC}
\end{flalign}
to the full subcategories  $\Alg(\MM)^{\ovr{\CC[W^{-1}]}}$ and 
$\Alg(\MM)^{\ovr{\CC}}$ of $\perp$-commutative functors.
Due to the universal property of localizations (see also \cite[Chapter 1]{Localization}),
the functor \eqref{eqn:Lasttmp} is a fully faithful embedding
and hence so is $\O_{L}^\ast$ by restriction to full subcategories.
\end{proof}
\begin{rem}
In the context of Example \ref{ex:LCQFTwithTS}, we interpret the left adjoint 
${\O_L}_!$ as the $W$-constantification/time-slicification functor
and the right adjoint $\O_L^\ast$ as the functor forgetting 
the $W$-constancy/time-slice axiom.
Proposition \ref{propo:timeslicificationproperties} has the following
interpretation: Take any quantum field theory
$A \in \Alg_{\O_{\ovr{\CC[W^{-1}]}}}(\MM)$ that does satisfy the
$W$-constancy axiom and forget this property by considering 
$\O_L^\ast(A) \in \Alg_{\O_{\ovr{\CC}}}(\MM)$. 
Applying the $W$-constantification functor 
determines a quantum field theory that is isomorphic
to $A$ via $\epsilon : {\O_L}_!\, \O_L^\ast (A) \to A$.
\end{rem}

\subsection{\label{subsec:subcats}Full orthogonal subcategories}
Let $\ovr{\DD} =(\DD,\perp)\in\OCat$ be an orthogonal category.
Let further $\CC\subseteq \DD$ be a full subcategory with embedding functor denoted
by $j:\CC\to\DD$.  We may equip $\CC$ with the pullback orthogonality relation
$j^\ast(\perp)$ from Lemma \ref{lem:pullpushorth}.
We call $\ovr{\CC} = (\CC,j^\ast(\perp))\in\OCat$ a {\em full orthogonal
subcategory} of $\ovr{\DD}$ and note that $j : \ovr{\CC}\to\ovr{\DD}$
is an orthogonal functor. By Proposition \ref{propo:functoriality} 
and Theorem \ref{theo:coloradjunctionalgebras}, we obtain an adjunction
\begin{flalign}\label{eqn:fullorthsubcatadjunction}
\xymatrix{
{\O_{j}}_! \,:\, \Alg_{\O_{\ovr{\CC}}}(\MM) ~\ar@<0.5ex>[r] & \ar@<0.5ex>[l]~ \Alg_{\O_{\ovr{\DD}}}(\MM) \,:\, \O_{j}^\ast
}\quad.
\end{flalign}
Recall from Example \ref{ex:LCQFTnoTS} that in this case the left adjoint 
should be interpreted as an extension functor of quantum field theories
defined on $\ovr{\CC}$ to theories on $\ovr{\DD}$. Adjunctions obtained 
from full orthogonal subcategory embeddings enjoy the following property.
\begin{propo}\label{propo:unitinjectionofcolor}
The adjunction \eqref{eqn:fullorthsubcatadjunction} exhibits 
$\Alg_{\O_{\ovr{\CC}}}(\MM)$ as a full {\em co}reflective subcategory of $\Alg_{\O_{\ovr{\DD}}}(\MM)$, i.e.\
the unit $\eta : \id_{\Alg_{\O_{\ovr{\CC}}}(\MM)} \to \O_{j}^\ast\, {\O_{j}}_!$ is a natural isomorphism.
\end{propo}
\begin{proof}
Given any $A \in \Alg_{\O_{\ovr{\CC}}}(\MM)$, we use Proposition \ref{propo:changeofcolorcoeq} to present
${\O_j}_!(A) \in \Alg_{\O_{\ovr{\DD}}}(\MM)$ as the reflexive coequalizer
\begin{flalign}
{\O_j}_!(A) ~ = ~ \colim \Big(\xymatrix@C=3em{
F_{\O_{\ovr{\DD}}}\, j_!\, F_{\O_{{\ovr{\CC}}}}(A) \ar@<0.5ex>[r]^-{\partial_0} \ar@<-0.5ex>[r]_-{\partial_1} ~&~ F_{\O_{{\ovr{\DD}}}}\, j_!(A)
}\Big)
\end{flalign}
in $\Alg_{\O_{\ovr{\DD}}}(\MM)$. Applying the right adjoint functor $\O_j^\ast$ and recalling that it
preserves reflexive coequalizers (cf.\ Remark \ref{rem:rightadjointofchangeofcolorpreservesreflexivecoeq}),
we obtain a natural isomorphism 
\begin{flalign}\label{eqn:tmpisofullorthsubcat}
\O_{j}^\ast\, {\O_{j}}_!(A) ~\cong~ \colim \Big(\xymatrix@C=3em{
\O_j^\ast\, F_{\O_{\ovr{\DD}}}\, j_! F_{\O_{\ovr{\CC}}}(A) \ar@<0.5ex>[r]^-{O_j^\ast (\partial_0)} \ar@<-0.5ex>[r]_-{O_j^\ast(\partial_1)} ~&~ \O_j^\ast\, F_{\O_{\ovr{\DD}}}\, j_!(A)
}\Big)
\end{flalign}
in $\Alg_{\O_{\ovr{\CC}}}(\MM)$. Because $j: \ovr{\CC}\to\ovr{\DD}$
is a full orthogonal subcategory embedding
we obtain that $\O_{\ovr{\DD}}\big(\substack{t\\\und{c}}\big) 
= \O_{\ovr{\CC}}\big(\substack{t\\\und{c}}\big)$,
for all $t \in \CC_0 \subseteq \DD_0$, $n \geq 0$, 
$\und{c} \in \CC_0^n \subseteq \DD_0^n$, and, 
for every $X \in \MM^{\CC_0}$ and $s\in \DD_0$, that
$j_! (X)_s \cong X_s$, for $s \in \CC_0 \subseteq \DD_0$,
and $j_! (X)_s \cong \emptyset$ else, where $\emptyset \in \MM$ 
denotes the initial object.
By a straightforward calculation using Theorem \ref{theo:freealgebraoperad} we then 
observe that the functor $\O_j^\ast\, F_{\O_{\ovr{\DD}}}\, j_! : \MM^{\CC_0}\to \Alg_{\O_{\ovr{\CC}}}(\MM)$
is naturally isomorphic to the free $\O_{\ovr{\CC}}$-algebra functor $F_{\O_{\ovr{\CC}}}$.
Applying this to the right-hand side of  \eqref{eqn:tmpisofullorthsubcat},
we obtain a natural isomorphism
\begin{flalign}
\O_{j}^\ast\, {\O_{j}}_!(A) ~ \cong ~ \colim \Big(\xymatrix@C=3em{
F_{\O_{\ovr{\CC}}}\, F_{\O_{\ovr{\CC}}}(A) \ar@<0.5ex>[r]^-{F_{\O_{\ovr{\CC}}}(\alpha)} \ar@<-0.5ex>[r]_-{\gamma_A} ~&~  F_{\O_{\ovr{\CC}}}(A)
}\Big)\quad.
\end{flalign}
By \cite[Lemma 4.3.3]{handbook2}, the right-hand side is naturally 
isomorphic to $A \in \Alg_{\O_{\ovr{\CC}}}(\MM)$, 
which shows that the functor $\O_{j}^\ast\, {\O_{j}}_!: 
\Alg_{\O_{\ovr{\CC}}}(\MM) \to \Alg_{\O_{\ovr{\CC}}}(\MM)$
is naturally isomorphic to the identity functor. 
This is sufficient to conclude that
the unit of the adjunction is a natural isomorphism, 
see e.g.\ \cite[Lemma~1.3]{JohnstoneMoerdijk}.
\end{proof}

Inspired by the application explained in Example \ref{ex:LCQFTnoTS}, 
we introduce the following concept.
\begin{defi}\label{def:jlocal}
We say that an object $A \in \Alg_{\O_{\ovr{\DD}}}(\MM)$ is {\em $j$-local} if
the corresponding component $\epsilon : {\O_{j}}_!\, \O_{j}^\ast (A) \to A$
of the counit of the adjunction \eqref{eqn:fullorthsubcatadjunction} 
is an isomorphism in $\Alg_{\O_{\ovr{\DD}}}(\MM)$. 
We denote the full subcategory of $j$-local objects
by $\Alg_{\O_{\ovr{\DD}}}(\MM)^{\text{$j$-loc}}$.
\end{defi}
Directly from Proposition \ref{propo:unitinjectionofcolor} 
and Definition \ref{def:jlocal} it follows that 
\begin{cor}\label{cor:jlocal}
\begin{itemize}
\item[(i)] For every $B \in \Alg_{\O_{\ovr{\CC}}}(\MM)$, the object ${\O_{j}}_! (B) \in \Alg_{\O_{\ovr{\DD}}}(\MM)$ is $j$-local.
\item[(ii)] The adjunction \eqref{eqn:fullorthsubcatadjunction} restricts to an adjoint equivalence
\begin{flalign}\label{eqn:fullorthsubcatadjunctionlocal} 
\xymatrix{
{\O_{j}}_! \,:\, \Alg_{\O_{\ovr{\CC}}}(\MM)  ~\ar@<0.8ex>[r]_-{\sim} & \ar@<0.8ex>[l] ~\Alg_{\O_{\ovr{\DD}}}(\MM)^{\text{$j$-loc}} \,:\, \O_{j}^\ast
}\quad.
\end{flalign}
\end{itemize}
\end{cor}

\begin{rem}
In the context of quantum field theory,
the full orthogonal subcategory $\ovr{\CC} \subseteq \ovr{\DD}$ 
should be interpreted as a subcategory of particularly ``nice'' spacetimes,
e.g.\ disks in Euclidean quantum field theory 
or diamonds in Lorentzian quantum field theory, cf.\ Example \ref{ex:LCQFTnoTS}.
A $j$-local object $A \in \Alg_{\O_{\ovr{\DD}}}(\MM)$ 
is a quantum field theory on the bigger spacetime category $\ovr{\DD}$
fully determined by its restriction $\O_j^\ast(A) \in \Alg_{\O_{\ovr{\CC}}}(\MM)$
to the subcategory of ``nice'' spacetimes $\ovr{\CC}$. In this sense, $j$-local
objects should be interpreted as quantum field theories that satisfy a local-to-global property,
which is similar to the one in factorization homology \cite{AyalaFrancis,LurieHA}.
Corollary \ref{cor:jlocal} states that the category of quantum field
theories that satisfy this local-to-global property is equivalent to the category
$\Alg_{\O_{\ovr{\CC}}}(\MM)$ of quantum field theories that are defined only on 
the subcategory $\ovr{\CC}$ of ``nice'' spacetimes. These techniques in particular
apply to spacetimes with boundaries, in which case such adjunctions
provide interesting insights into boundary conditions
for quantum field theories, see \cite{BDSboundary}.
\end{rem}

\subsection{\label{subsec:equivalences}Orthogonal equivalences}
We introduce a suitable notion of equivalence
$F:\ovr{\CC}\to \ovr{\DD}$ between orthogonal categories. We then show that the
adjunction \eqref{eqn:generalorthadjunction} induced by an orthogonal equivalence $F$
is an adjoint equivalence between the associated categories of algebras. 
In the terminology of \cite{KapranovManin}, this means that the $\Op$-morphism
$\O_F: \O_{\ovr{\CC}} \to \O_{\ovr{\DD}}$ is a Morita equivalence.

\begin{defi}\label{def:ortheq}
An orthogonal functor $F : \ovr{\CC}\to\ovr{\DD}$ 
is called an {\em orthogonal equivalence} if the following two properties hold true:
(1)~$F:\CC\to\DD$ is an equivalence of small categories, i.e.\ a fully faithful and
essentially surjective functor, and (2)~$F^\ast(\perp_\DD)=\perp_\CC$.
\end{defi}

\begin{theo}\label{thm:Morita}
Let $F: \ovr{\CC} \to \ovr{\DD}$ be an orthogonal equivalence. Then the induced adjunction 
\eqref{eqn:generalorthadjunction} is an adjoint equivalence
\begin{flalign}
\xymatrix{
{\O_{F}}_! \,:\, \Alg_{\O_{\ovr{\CC}}}(\MM) ~\ar@<0.8ex>[r]_-{\sim} & \ar@<0.8ex>[l]~ \Alg_{\O_{\ovr{\DD}}}(\MM) \,:\, \O_{F}^\ast
}\quad.
\end{flalign}
\end{theo}
\begin{proof}
Let us first consider the special case where the functor $F:\CC\to\DD$ is also injective on objects.
Then the unit of the adjunction \eqref{eqn:generalorthadjunction}
is a natural isomorphism because of Proposition \ref{propo:unitinjectionofcolor}.
We now show that the counit is a natural isomorphism too. Notice that there is the following chain of 
natural isomorphisms
\begin{flalign}\label{eqn:tmpchainisoequivalence}
{\O_F}_!\, \O_F^\ast ~\cong~ {p_{\ovr{\DD}}}_!\, \Lan_F\, p_{\ovr{\CC}}^\ast\, \O_F^\ast 
~=~ {p_{\ovr{\DD}}}_!\, \Lan_F\, F^\ast\, p_{\ovr{\DD}}^\ast~ \cong~ {p_{\ovr{\DD}}}_!\, p_{\ovr{\DD}}^\ast 
~\cong~ \id_{\Alg_{\O_{\ovr{\DD}}}(\MM)} \quad.
\end{flalign}
In the first step we used Proposition \ref{propo:causalcorrect} 
and in the second step we used that the square of right adjoints in 
\eqref{eqn:LanopLandiagram} commutes. In step three we used
that $F:\CC\to\DD$ is an equivalence of small categories, which 
implies that the pullback functor $F^\ast$ is fully faithful
and hence that the counit $\Lan_F\, F^\ast \to \id_{\Alg(\MM)^\DD}$
of the adjunction \eqref{eqn:Lanadjunction} is a natural isomorphism. The last step follows from 
Lemma \ref{lem:pbarCreflective}. The dual of \cite[Lemma~1.3]{JohnstoneMoerdijk} 
allows us to conclude from \eqref{eqn:tmpchainisoequivalence} 
that the counit $\epsilon$ is a natural isomorphism.
\sk

The general case can be reduced to the special case above by the following argument:
Let us choose a skeleton $\CC^\prime \subseteq \CC$, with embedding functor
denoted by $j:\CC^\prime\to\CC$, and define
$\ovr{\CC}^\prime := (\CC^\prime,j^\ast(\perp_\CC))$. Note that
$\ovr{\CC}^\prime$ is a full orthogonal subcategory of $\ovr{\CC}$.
One easily confirms that 
both $j :\ovr{\CC}^\prime\to\ovr{\CC}$ and $F\,j : \ovr{\CC}^\prime\to \ovr{\DD}$
are orthogonal equivalences that are injective on objects.
Our results above then imply that both $F\, j$ and $j$ induce
adjoint equivalences between the associated categories of algebras.
To complete the proof, we notice that 
the 2-out-of-3 property of equivalences of categories and 
Remark \ref{rem:coloradjunctioncomposition}
implies that also $F$ induces an adjoint equivalence. 
\end{proof}

\begin{rem}
The practical relevance of this result is the following:
Recall from the examples in Section \ref{subsec:examples}
that one is often interested in studying quantum field theories 
which are defined on an orthogonal category that is only {\em essentially} small. 
To avoid set theoretic issues, one has to replace such orthogonal categories
by equivalent small orthogonal categories, whose choice is typically not unique.
Different choices in general define non-isomorphic colored operads which,
however, are Morita-equivalent because of Theorem \ref{thm:Morita},
i.e.\ the associated categories of algebras are naturally equivalent.
The practical implication is that the category of quantum field theories 
does not depend on the choice of a small model for the orthogonal category of interest. 
\end{rem}

\subsection{\label{subsec:orbifold}Right adjoints and orbifoldization}
Given an orthogonal functor $F:\ovr{\CC}\to\ovr{\DD}$,
our focus so far was on the induced adjunction \eqref{eqn:generalorthadjunction}
where the pullback $\O_F^\ast$ is a right adjoint functor and ${\O_F}_!$ is its left adjoint.
Forgetting for the moment the orthogonality relations on our categories, this reduces to
the adjunction \eqref{eqn:Lanadjunction} obtained by left Kan extension.
Because the underlying base category $\MM$ is by hypothesis also complete,
there exists another adjunction (obtained by right Kan extension)
\begin{flalign}\label{eqn:Ranadjunction}
\xymatrix{
F^\ast \,:\, \Alg(\MM)^\DD~\ar@<0.5ex>[r]&\ar@<0.5ex>[l]  ~\Alg(\MM)^\CC \,:\, \Ran_F
}\quad,
\end{flalign}
where $F^\ast$ is the left adjoint. It is natural to ask whether also the 
functor $\O_F^\ast: \Alg_{\O_{\ovr{\DD}}}(\MM) \to \Alg_{\O_{\ovr{\CC}}}(\MM)$ 
admits a right adjoint. In general, this is not the case due to the following
\begin{ex}
Consider the category $\{\ast\}$ consisting of one object $\ast$ and its identity morphism $\id_\ast$.
On this category there exist two different orthogonality relations $\perp_{\mathrm{min}} = \emptyset$
and $\perp_{\mathrm{max}} = \{(\id_\ast,\id_\ast)\}$. 
By Theorem \ref{theo:AlgoverOCCperp}, we obtain
\begin{flalign}
\Alg_{\O_{(\{\ast\},\perp_{\mathrm{min}})}}(\MM)~\cong~\Alg(\MM)~~,\quad
\Alg_{\O_{(\{\ast\},\perp_{\mathrm{max}})}}(\MM)~\cong~\CAlg(\MM)\quad,
\end{flalign}
where $\CAlg(\MM)$ is the category of commutative, associative and unital algebras in $\MM$.
The orthogonal functor
$(\{\ast\},\perp_{\mathrm{min}}) \to (\{\ast\},\perp_{\mathrm{max}})$
induces the adjunction
\begin{flalign}
\xymatrix{
\mathrm{Ab} \,:\, \Alg(\MM)~\ar@<0.5ex>[r]&\ar@<0.5ex>[l]  ~\CAlg(\MM) \,:\, U
}\quad.
\end{flalign}
The right adjoint $U$ is the functor forgetting commutativity and the
left adjoint $\mathrm{Ab}$ is the abelianization of algebras in $\MM$.
Since $U$ fails to preserve coproducts, it can not be a left adjoint functor.
This implies that $\O_F^\ast : \Alg_{\O_{\ovr{\DD}}}(\MM)\to \Alg_{\O_{\ovr{\CC}}}(\MM)$ 
does not always admit a right adjoint.
\end{ex}

We now consider a special situation where it turns out that 
the functor  $\O_F^\ast$ does admit a right adjoint. The motivation for this scenario
comes from {\em orbifoldization}, which is the procedure
of assigning to quantum field theories with group (or groupoid) actions
their corresponding invariants \cite{Dijkgraaf}. 
Such constructions were studied by \cite{BSfiberedingroupoids,BSWhoRan} 
in the context of algebraic quantum field theory 
and by \cite{SW} in the context of topological quantum field theory. 
It is important to stress that the procedure of 
taking invariants is formalized by categorical limits
and hence is related to right adjoints of the functor  $\O_F^\ast$.
\sk

Our scenario is as follows: Let $\ovr{\DD} = (\DD,\perp)$ 
be an orthogonal category and $F : \CC\to\DD$ a category fibered in 
groupoids over $\DD$, see e.g.\  \cite{BSfiberedingroupoids} for a definition. 
We equip $\CC$ with the pullback orthogonality relation $F^\ast(\perp)$
and call the resulting orthogonal functor 
$F :  \ovr{\CC}:=(\CC,F^\ast(\perp)) \to\ovr{\DD}$ 
an {\em orthogonal category fibered in groupoids}.
\begin{propo}
Let $F: \ovr{\CC} \to \ovr{\DD}$ be an orthogonal category fibered in groupoids. 
Then the pullback functor 
$\O_F^\ast: \Alg_{\O_{\ovr{\DD}}}(\MM) \to \Alg_{\O_{\ovr{\CC}}}(\MM)$ 
has a right adjoint, i.e.\ there is an adjunction 
\begin{flalign}
\xymatrix{
\O_F^\ast \,:\, \Alg_{\O_{\ovr{\DD}}}(\MM) ~\ar@<0.5ex>[r] & \ar@<0.5ex>[l]~ \Alg_{\O_{\ovr{\CC}}}(\MM) \,:\, {\O_F}_\ast
}\quad.
\end{flalign}
We call the right adjoint $ {\O_F}_\ast$ the orbifoldization functor.
\end{propo}
\begin{proof}
In \cite[Theorem~4.3]{BSfiberedingroupoids} it was shown that under our hypotheses 
the right Kan extension $\Ran_F : \Alg(\MM)^\CC \to \Alg(\MM)^\DD$
preserves $\perp$-commutativity. Using Lemma \ref{lem:pbarCreflective},
this implies the existence of a unique functor 
${\O_F}_\ast: \Alg_{\O_{\ovr{\CC}}}(\MM) \to \Alg_{\O_{\ovr{\DD}}}(\MM)$ 
such that $\Ran_F \, p_{\ovr{\CC}}^\ast = p_{\ovr{\DD}}^\ast\,  {\O_F}_\ast$. 
It remains to show that ${\O_F}_\ast$ is the right adjoint of $\O_F^\ast$.
Given any $A \in \Alg_{\O_{\ovr{\CC}}}(\MM)$ 
and $B \in \Alg_{\O_{\ovr{\DD}}}(\MM)$, there is the
following chain of natural bijections of Hom-sets
\begin{flalign}
\nn \Alg_{\O_{\ovr{\DD}}}(\MM) \big(B, {\O_F}_\ast (A)\big) 
&~\cong~ \Alg(\MM)^\DD \big(p_{\ovr{\DD}}^\ast (B) , p_{\ovr{\DD}}^\ast \, {\O_F}_\ast (A)\big) \\
\nn &~\cong~ \Alg(\MM)^\DD \big( p_{\ovr{\DD}}^\ast (B), \Ran_F\, p_{\ovr{\CC}}^\ast (A)\big) \\
\nn &~\cong~ \Alg(\MM)^\CC \big( F^\ast\, p_{\ovr{\DD}}^\ast (B), p_{\ovr{\CC}}^\ast (A)\big) \\
\nn &~\cong~ \Alg(\MM)^\CC \big( p_{\ovr{\CC}}^\ast\, \O_F^\ast (B), p_{\ovr{\CC}}^\ast (A)\big) \\ 
&~\cong~ \Alg_{\O_{\ovr{\CC}}}(\MM) \big( \O_F^\ast (B),  A \big) \quad.
\end{flalign} 
In the first and last step we used Lemma \ref{lem:pbarCreflective}
and in step four we used that the square formed by the right adjoints in 
\eqref{eqn:LanopLandiagram} commutes. This proves that ${\O_F}_\ast$ 
is the right adjoint of $\O_F^\ast$.
\end{proof}

%%%%%%%%%%%%%%%%%%%%%%%%%%%%%%%%%%%%%%%%%%%%%%%%
%%%%%%%%%%%%%%%%%%%%%%%%%%%%%%%%%%%%%%%%%%%%%%%%

\section{\label{sec:QFTconstructions}Comparison to Fredenhagen's universal algebra}
In \cite{Fre1,Fre2,Fre3}, Fredenhagen studied extensions of
quantum field theories that are defined only on certain open subsets 
$U\subseteq M$ of a spacetime manifold $M$ to the whole of $M$.
It was later recognized by Lang in his PhD thesis \cite{Lang} that this extension
may be formalized as a left Kan extension of the functor underlying the quantum field theory.
In our notation and language, Fredenhagen's universal algebra construction
can be formalized as follows: Let $\ovr{\DD} = (\DD,\perp)\in\OCat$ be an orthogonal category.
Let $\ovr{\CC} = (\CC,j^\ast(\perp))$ be a full orthogonal subcategory
(cf.\ Section \ref{subsec:subcats}) with orthogonal embedding functor denoted by 
$j:\ovr{\CC}\to\ovr{\DD}$. Fredenhagen's universal algebra construction \cite{Fre1,Fre2,Fre3,Lang} 
assigns to a quantum field theory $A \in \Alg_{\O_{\ovr{\CC}}}(\MM)$ on $\ovr{\CC}$ the 
algebra-valued functor
\begin{flalign}\label{eqn:Fredenhagen}
\Lan_j \, p_{\ovr{\CC}}^\ast (A) \in \Alg(\MM)^\DD
\end{flalign}
on the category $\DD$. Notice that the construction \eqref{eqn:Fredenhagen} 
consists of two steps: First, one applies the functor 
$p_{\ovr{\CC}}^\ast: \Alg_{\O_{\ovr{\CC}}}(\MM) \to \Alg(\MM)^\CC$
that forgets $\perp$-commutativity, assigning to the quantum field theory
$A \in \Alg_{\O_{\ovr{\CC}}}(\MM)$ its underlying algebra-valued functor $p_{\ovr{\CC}}^\ast (A) \in \Alg(\MM)^\CC$ 
on $\CC$. In the second step this underlying functor is extended from $\CC$ to $\DD$ via left Kan extension
along the embedding functor $j:\CC\to\DD$.
\sk

A potential weakness of this construction is that it is unclear whether the extended
functor \eqref{eqn:Fredenhagen} satisfies the $\perp$-commutativity axiom on 
$\ovr{\DD}$, i.e.\ whether it is a bona fide quantum field theory in the sense of
an object in $\Alg_{\O_{\ovr{\DD}}}(\MM)$. This weakness
does not appear in our operadic construction explained in Section \ref{subsec:subcats}.
Concretely, instead of using \eqref{eqn:Fredenhagen} to
extend the quantum field theory $A \in \Alg_{\O_{\ovr{\CC}}}(\MM)$ from $\ovr{\CC}$ to $\ovr{\DD}$,
we use the left adjoint in \eqref{eqn:fullorthsubcatadjunction} to assign
\begin{flalign}\label{eqn:Our}
{\O_{j}}_!(A) \in \Alg_{\O_{\ovr{\DD}}}(\MM) \quad.
\end{flalign}
By construction, our extended theory satisfies the $\perp$-commutativity axiom
on $\ovr{\DD}$, i.e.\ it is a bona fide quantum field theory. 
The aim of this section is to compare our construction \eqref{eqn:Our}
to the construction \eqref{eqn:Fredenhagen} of Fredenhagen and Lang.
Our first result is that whenever \eqref{eqn:Fredenhagen} satisfies 
the $\perp$-commutativity axiom on $\ovr{\DD}$, then it agrees
with our construction \eqref{eqn:Our}.
\begin{propo}\label{propo:comparison1}
Let $A \in \Alg_{\O_{\ovr{\CC}}}(\MM)$ be such that 
$\Lan_j \, p_{\ovr{\CC}}^\ast (A) \in \Alg(\MM)^\DD$
is $\perp$-commutative on $\ovr{\DD}$. Then there exists an isomorphism
\begin{flalign}
\Lan_j \, p_{\ovr{\CC}}^\ast (A) ~\cong~p_{\ovr{\DD}}^\ast \,{\O_j}_!(A)
\end{flalign}
in $\Alg(\MM)^\DD$.
\end{propo}
\begin{proof}
By Lemma \ref{lem:pbarCreflective}, we know that $\Alg_{\O_{\ovr{\DD}}}(\MM)$ is a full reflective
subcategory of $\Alg(\MM)^\DD$. Because $\Lan_j \, p_{\ovr{\CC}}^\ast (A) \in \Alg(\MM)^\DD$
satisfies by hypothesis the $\perp$-commutativity axiom, it follows that there exists
$\widehat{A} \in \Alg_{\O_{\ovr{\CC}}}(\MM)$ such that $ p_{\ovr{\DD}}^\ast(\widehat{A})\cong \Lan_j \, p_{\ovr{\CC}}^\ast (A) $.
Applying ${p_{\ovr{\DD}}}_!$ we obtain
\begin{flalign}
\widehat{A} ~\cong~{p_{\ovr{\DD}}}_!\, p_{\ovr{\DD}}^\ast(\widehat{A})~\cong~
 {p_{\ovr{\DD}}}_!\, \Lan_j \, p_{\ovr{\CC}}^\ast (A) ~\cong~{\O_{j}}_!(A) 
\end{flalign}
in $\Alg_{\O_{\ovr{\DD}}}(\MM)$, where in the first step we used that the counit 
of the adjunction ${p_{\ovr{\DD}}}_! \dashv p_{\ovr{\DD}}^\ast$ 
is a natural isomorphism (cf.\ Lemma \ref{lem:pbarCreflective}) 
and in the last step we used Proposition \ref{propo:causalcorrect}. 
\end{proof}

It thus remains to understand whether
\eqref{eqn:Fredenhagen} does
satisfy the $\perp$-commutativity axiom.
Our strategy to address this question is to compute
explicitly the functor \eqref{eqn:Fredenhagen}
by using the operadic techniques from Section \ref{subsec:Alg}.
To simplify the presentation, we assume that the underlying base category $\MM$
is concrete and that the monoidal unit $\1 \not\cong \emptyset$ is not isomorphic
to the initial object. For example, 
we could take $\MM=\Vec_\bbK$. 
This allows us to think of the objects in $\MM$ as sets with additional 
structures and of the morphisms as structure preserving functions.
In particular, we can perform element-wise computations.
(Using the concept of generalized elements, there is no need to assume that
$\MM$ is concrete. However, we decided to add this reasonable assumption
to make our presentation more transparent.)
\sk

The problem of computing left Kan extensions can be addressed within 
our operadic formalism. Recalling from Remark \ref{rem:AlgoverOCC}
that $\Alg(\MM)^\EE \cong \Alg_{\O_{\EE}}(\MM)$ for any small category $\EE$,
we may describe algebra-valued functors in terms of algebras
over our auxiliary operad $\O_{\EE}$, see Definition \ref{def:OCC}. 
The left Kan extension $\Lan_j : \Alg(\MM)^\CC\to \Alg(\MM)^\DD$
is then identified with the left adjoint of the adjunction
\begin{flalign}
\xymatrix{
{\O_{j}}_! \,:\, \Alg_{\O_{\CC}}(\MM) ~\ar@<0.5ex>[r] & \ar@<0.5ex>[l]~ \Alg_{\O_{\DD}}(\MM) \,:\, \O_{j}^\ast
}\quad,
\end{flalign}
which is induced by applying Theorem \ref{theo:coloradjunctionalgebras}
to the $\Op$-morphism $\O_j : \O_{\CC }\to \O_{\DD }$ 
between our auxiliary operads (cf.\ Proposition \ref{propo:functoriality}).
Using Propositions \ref{propo:changeofcolorcoeq} 
and \ref{propo:colimitsinAlg}, we obtain that the 
left Kan extension of an algebra-valued functor $B\in \Alg(\MM)^\CC$
can be computed by the point-wise reflexive coequalizer
\begin{flalign}\label{eqn:LanjBd}
\Lan_j(B)_d ~=~\colim\Big(\xymatrix{
F_{\O_\DD}\, j_!\, F_{\O_\CC}\,(B)_d \ar@<0.5ex>[r]^-{\partial_0}  \ar@<-0.5ex>[r]_-{\partial_1} ~&~ F_{\O_\DD}\, j_!(B)_d
}\Big)
\end{flalign}
in $\MM$. Explicitly, using \eqref{eqn:freealgformula} 
and that $j:\CC\to\DD$ is an inclusion on the sets of objects, we obtain
\begin{flalign}
F_{\O_\DD}\, j_!(B)_d ~\cong~ \int^{\und{c}}\O_\DD\big(\substack{d\\\und{c}}\big)\otimes B_{\und{c}} ~\cong~ \coprod\limits_{\und{c}} \DD(\und{c},d)\otimes B_{\und{c}}\quad,
\end{flalign}
where we denote the objects of the subcategory $\CC\subseteq \DD$ by $c$'s
and generic objects of $\DD$ by $d$'s 
and $\DD(\und{c},d) := \prod_i \DD(c_i,d)$. 
Notice that in the last step we used that the equivalence classes 
induced by the coend admit unique representatives, that we denote by  
\begin{flalign}
\und{g}\otimes\und{b} ~\in~ F_{\O_\DD}\, j_!(B)_d \quad,
\end{flalign}
for $\und{g} \in \DD(\und{c},d)$ and $\und{b} \in B_{\und{c}}$. 
The algebra structure on $F_{\O_\DD}\, j_!(B)_d$ is given by
\begin{flalign}\label{eqn:monoidLan}
\mu_d\big((\und{g}\otimes\und{b})\otimes (\und{g}^\prime \otimes\und{b}^\prime)\big)\, =\, (\und{g},\und{g}^\prime)\otimes \und{b}\otimes\und{b}^\prime~~,\quad 1_d \, = \, \ast\in \DD(\emptyset,d)\quad,
\end{flalign}
where $(\und{g},\und{g}^\prime) := (g_1,\dots,g_n , g_1^\prime, \dots, g_m^\prime)$ is the concatenation
of $\und{g}$ and $\und{g}^\prime$. 
The other object $F_{\O_\DD}\, j_!\, F_{\O_\CC}\, (B)_d$ 
and the morphisms $\partial_0$, $\partial_1$ in \eqref{eqn:LanjBd}
can be computed similarly. The result is
\begin{lem}\label{lem:Lanexplicit}
The functor $\Lan_j(B) : \DD\to\Alg(\MM)$ has the following explicit description:
To an object $d\in\DD$, it assigns the algebra 
\begin{subequations}\label{eqn:Lanexplicitrelations}
\begin{flalign}
\Lan_j(B)_d ~=~ F_{\O_\DD}\, j_!(B)_d\big/\!\sim
\end{flalign}
given by implementing the equivalence relation
\begin{flalign}
\und{g}(\und{f}_1,\dots,\und{f}_n)\otimes\und{b}_1\otimes\cdots\otimes \und{b}_n~\sim~
\und{g}\otimes B(\und{f}_1)(\und{b}_1)\otimes \cdots\otimes B(\und{f}_n)(\und{b}_n)\quad,
\end{flalign}
\end{subequations}
for all $\und{g} \in \DD(\und{c},d)$ with $\und{c} = (c_1,\ldots,c_n)$ 
and $\und{f}_i\in \CC(\und{c}_i,c_i)$, $\und{b}_i \in B_{\und{c}_i}$, 
for  $i=1,\dots,n$.
Here  $\und{g}(\und{f}_1,\dots,\und{f}_n)$ is defined in \eqref{eqn:CCcomp}
and the $n$-fold product $B(\und{f})(\und{b}) := B(f_1)(b_1)\, \cdots\, B(f_n)(b_n)$
is defined by the algebra-valued functor $B: \CC\to\Alg(\MM)$. 
The algebra structure in \eqref{eqn:monoidLan} descends to \eqref{eqn:Lanexplicitrelations}.
Furthermore, to a $\DD$-morphism $h : d\to d^\prime$, 
$\Lan_j(B)$ assigns the $\Alg(\MM)$-morphism 
\begin{flalign}
\Lan_j(B)(h) \, :\, \Lan_j(B)_d  ~\longrightarrow~\Lan_j(B)_{d^\prime} 
~~,\quad \big[\und{g}\otimes\und{b}\big]~\longmapsto ~\big[h(\und{g})\otimes \und{b}\big]\quad.
\end{flalign}
\end{lem}

We can now answer the question when \eqref{eqn:Fredenhagen} 
satisfies the $\perp$-commutativity axiom. For this it will be 
useful to introduce the following terminology.
\begin{defi}\label{defi:jclosed}
We say that an object $d\in\DD$ is {\em $j$-closed} if, for every
pair of orthogonal morphisms $(g_1 : c_1\to d)\perp (g_2: c_2\to d)$
with target $d$ and sources $c_1,c_2\in\CC$, 
there exist a morphism $g:c\to d$ with $c\in\CC$ and 
$(f_1:c_1\to c)\perp(f_2:c_2\to c)$ such that the diagram
\begin{flalign}
\xymatrix@C=3em@R=2em{
& d &\\
c_1 \ar[ru]^-{g_1} \ar@{-->}[r]_-{f_1} & c \ar@{-->}[u]^-{g} & c_2 \ar[lu]_-{g_2} \ar@{-->}[l]^-{f_2}
}
\end{flalign}
in $\DD$ commutes.
\end{defi}
\begin{theo}\label{theo:jclosed}
$\Lan_j \, p_{\ovr{\CC}}^\ast (A)\in \Alg(\MM)^\DD$ is $\perp$-commutative over $d\in\DD$,
for all $A \in \Alg_{\O_{\ovr{\CC}}}(\MM)$, if and only if $d$ is $j$-closed.
\end{theo}
\begin{proof}
Let us first prove the direction ``$\Leftarrow$'':
We have to show that $\widetilde A := \Lan_j \, p_{\ovr{\CC}}^\ast (A)\in \Alg(\MM)^\DD$
is $\perp$-commutative over $d$, i.e.\ that, 
for all $(h_1: d_1\to d)\perp (h_2:d_2\to d)$ with target $d$,  
\begin{flalign}
\mu_d\Big( \widetilde A(h_1)\big(\tilde a_1\big)\otimes  
\widetilde A(h_2)\big(\tilde a_2\big) \Big)
~=~\mu_d^\op\Big( \widetilde A(h_1)\big(\tilde a_1\big)\otimes  
\widetilde A(h_2)\big(\tilde a_2\big) \Big)\quad,
\end{flalign}
for all $\tilde a_i\in \widetilde A_{d_i}$, $i=1,2$.
Using Lemma \ref{lem:Lanexplicit}, this condition explicitly reads as
\begin{flalign}\label{eqn:tmpperpcomelements}
\Big[\big(h_1(\und{g}_1) , h_2(\und{g}_2)\big)\otimes \und{a}_1\otimes\und{a}_2\Big]
~=~\Big[\big(h_2(\und{g}_2) , h_1(\und{g}_1)\big)\otimes \tau(\und{a}_1\otimes\und{a}_2)\Big]\quad,
\end{flalign}
for all $\und{g}_i\otimes \und{a}_i \in \coprod_{\und{c}} \DD(\und{c},d_i)\otimes A_{\und{c}}$, 
$i=1,2$,
where $\tau$ is the symmetric braiding of $\MM$. It is sufficient to prove \eqref{eqn:tmpperpcomelements}
for elements $g_i\otimes a_i\in \DD(c_i,d_i)\otimes A_{c_i}$, $i=1,2$,
of length $1$; the general case follows from this by iteration.
Because by assumption $h_1\perp h_2$, it follows from composition stability
of $\perp$ that $(h_1\,g_1: c_1\to d) \perp (h_2\,g_2:c_2\to d)$.
Using further that $d$ is by hypothesis $j$-closed, we find
$g: c\to d$ in $\DD$ and $(f_1:c_1\to c) \perp (f_2:c_2\to c)$ in $\CC$,
such that $(h_1\,g_1,h_2\,g_2) = (g\,f_1,g\,f_2)$. Using the relations in \eqref{eqn:Lanexplicitrelations},
we obtain
\begin{flalign}
\nn\Big[\big(h_1\,g_1 , h_2\, g_2\big)\otimes a_1\otimes a_2\Big]~&=~
\Big[g (f_1 ,f_2)\otimes a_1\otimes a_2\Big]~=~
\Big[g \otimes \mu_c \big( A(f_1)(a_1)\otimes A(f_2)(a_2)\big)\Big]\\
\nn ~&=~\Big[g \otimes \mu_c^\op \big(A(f_1)(a_1)\otimes A(f_2)(a_2)\big)\Big]~=~
\Big[g (f_2 ,f_1)\otimes \tau(a_1\otimes a_2)\Big]\\
~&=~ \Big[\big(h_2\,g_2 , h_1\, g_1\big)\otimes \tau(a_1\otimes a_2)\Big]\quad,
\end{flalign}
where in the third step we used that $A$ is $\perp$-commutative on $\ovr{\CC}$.
\sk

We now prove ``$\Rightarrow$'': Let $(g_1: c_1 \to d) \perp (g_2: c_2 \to d)$ be any
orthogonal pair of morphisms. Our strategy is to construct 
an object $A \in \Alg_{\O_{\ovr{\CC}}}(\MM)$ such that $\perp$-commutativity
of $\Lan_j \, p_{\ovr{\CC}}^\ast (A)$ 
implies the existence of a factorization as in Definition \ref{defi:jclosed}.
We define $X \in \MM^{\CC_0}$ as follows: 
If $c_1=c_2$, we set $X_{c_1} := I \sqcup I$ and $X_c := \emptyset$, 
for $c \in \CC_0 \setminus \{c_1\}$; otherwise, we set $X_{c_1} := I$, 
$X_{c_2} := I$ and $X_c := \emptyset$, for $c \in \CC_0 \setminus \{c_1,c_2\}$.
Because we assume that $\1 \not\cong \emptyset$, 
there exists an element $x_1\in X_{c_1}$ and an 
element $x_2\in X_{c_2}$, such that $x_1$ and $x_2$ are different in the case of $c_1=c_2$.
Let us consider the free $\O_{\ovr{\CC}}$-algebra 
$A := F_{\O_{\ovr{\CC}}}(X) \in \Alg_{\O_{\ovr{\CC}}}(\MM)$
and note that $x_1\in A_{c_1}$ and $x_2\in A_{c_2}$ are generators.
Because $\Lan_j\, p_{\ovr{\CC}}^\ast (A)$ is by hypothesis 
$\perp$-commutative over $d\in \DD$, it follows that
\begin{flalign}
\Big[ (g_1,g_2) \otimes x_1 \otimes x_2 \Big] 
~=~ \Big[ (g_2,g_1) \otimes \tau(x_1 \otimes x_2) \Big] \quad,
\end{flalign}
where $(g_1: c_1 \to d) \perp (g_2: c_2 \to d)$ is the given orthogonal pair of morphisms.
Using the equivalence relation \eqref{eqn:Lanexplicitrelations}
and that $x_1$, $x_2$ are two distinct generators of the free $\O_{\ovr{\CC}}$-algebra $A$, 
one observes that this equality of equivalence classes can only hold true if
both sides admit a representative of length $1$, i.e.\ with only one tensor factor in $A$. Hence,
there must exist a factorization $(g_1,g_2) = (g\,f_1,g\,f_2)$
with $f_i : c_i\to c$ in $\CC$, for $i=1,2$, and $g : c\to d$ in $\DD$.
This yields the equality
\begin{flalign}
\Big[ g \otimes \mu_c \big( A(f_1)(x_1) \otimes A(f_2)(x_2) \big) \Big] 
~=~ \Big[ g \otimes \mu_c^\op \big( A(f_1)(x_1) \otimes A(f_2)(x_2) \big) \Big] \quad,
\end{flalign}
which ensures the existence of a further factorization $g = g^\prime\,f^\prime$
with $f^\prime : c\to c^\prime$ in $\CC$ and $g^\prime : c^\prime\to d$
in $\DD$ such that $(f^\prime\, f_1 , f^\prime\, f_2\big)\in j^\ast(\perp)$.
\end{proof}

Let us denote by $\DD^\prime\subseteq \DD$ the full subcategory
of $j$-closed objects. Notice that every object $c\in\CC$ is $j$-closed, 
hence $\CC\subseteq \DD^\prime$.
Equipping $\DD^\prime$ with the pullback orthogonality relation,
we obtain a factorization of $j : \ovr{\CC}\to\ovr{\DD}$ into two
full orthogonal subcategory embeddings $j^\prime : \ovr{\CC}\to\ovr{\DD}^\prime$
and $j^{\prime\prime} : \ovr{\DD}^\prime\to\ovr{\DD}$.
Combining Theorem \ref{theo:jclosed} and Proposition \ref{propo:comparison1},
we obtain
\begin{cor}
Let $\ovr{\DD}^\prime$ be the full orthogonal subcategory of $j$-closed objects in $\ovr{\DD}$
and consider the full orthogonal subcategory embedding $j^\prime : \ovr{\CC}\to\ovr{\DD}^\prime$.
Then there exists a natural isomorphism
\begin{flalign}
\Lan_{j^\prime}\,p_{\ovr{\CC}}^\ast ~\cong~ p_{\ovr{\DD}^\prime}^\ast\, {\O_{j^\prime}}_!\quad.
\end{flalign}
\end{cor}
We conclude this section by providing examples and counterexamples
of $j$-closed objects in the context of the examples discussed in Section 
\ref{subsec:examples}.
\begin{ex}
In the context of Example \ref{ex:LCQFTnoTS}, consider the
full orthogonal subcategory $j : \ovr{\Loc}_\diamond\to\ovr{\Loc}$ 
of diamond spacetimes $\mathbb{U}$, whose underlying manifold
is diffeomorphic to $\bbR^m$. This scenario has been studied in \cite{Lang}.
We first notice that every disconnected spacetime $\mathbb{M}\in\Loc$
is {\em not} $j$-closed: Two embeddings $g_1 : \mathbb{U}_1\to \mathbb{M}$
and $g_2 : \mathbb{U}_2\to \mathbb{M}$ of diamonds into different connected components of $\mathbb{M}$
are clearly orthogonal, $g_1\perp g_2$, however they do not factorize through
a common diamond $g: \mathbb{U} \to \mathbb{M}$. 
Hence, Fredenhagen's universal algebra \eqref{eqn:Fredenhagen}
in general fails to produce functors that are $\perp$-commutative over disconnected
spacetimes and our construction \eqref{eqn:Our} solves this issue.
On the other hand, objects $\mathbb{M}\in\Loc$ whose underlying manifold
is diffeomorphic to $\bbR\times\mathbb{S}^{m-1}$ are $j$-closed:
Given two causally disjoint embeddings $g_1 : \mathbb{U}_1\to \mathbb{M}$ 
and $g_2 : \mathbb{U}_2\to \mathbb{M}$ of diamonds, $g_1\perp g_2$,
they factorize through the diamond inclusion $g : \mathbb{U}\to\mathbb{M}$,
where $\mathbb{U}$ is defined by restricting $\mathbb{M}$
to the globally hyperbolic open subset $U = M\setminus J_{\mathbb{M}}(\{x\}) \cong \bbR^m$
for some $x\in M \setminus J_{\mathbb{M}}(g_1(U_1)\cup g_{2}(U_2))$.
A complete characterization of the $j$-closed objects in $\ovr{\Loc}$ seems to be
rather complicated and is beyond the scope of this article. 
Similar conclusions can be obtained also in the context 
of chiral conformal and Euclidean field theories. 
\end{ex}

%%%%%%%%%%%%%%%%%%%%%%%%%%%%%%%%%%%%%%%%%%%%%%%%
%%%%%%%%%%%%%%%%%%%%%%%%%%%%%%%%%%%%%%%%%%%%%%%%

\section*{Acknowledgments}
We are very grateful to Ulrich Bunke for encouraging us to
develop an operadic framework for algebraic quantum field theory. 
We also would like to thank Claudia Scheimbauer and 
Christoph Schweigert for useful technical comments and suggestions.
The work of M.B.\ is supported by a research grant funded by 
the Deutsche Forschungsgemeinschaft (DFG, Germany). 
A.S.\ gratefully acknowledges the financial support of 
the Royal Society (UK) through a Royal Society University 
Research Fellowship, a Research Grant and an Enhancement Award. 
L.W.\ is supported by the RTG 1670 ``Mathematics inspired 
by String Theory and Quantum Field Theory''.

%%%%%%%%%%%%%%%%%%%%%%%%

\end{document}